\documentclass[preprint, review, 12pt]{article}          
   
\usepackage{mathrsfs}  
\usepackage{amsmath}
\usepackage{amssymb}     
\usepackage{graphicx}               
\usepackage{latexsym}              
\usepackage{amsfonts}  
\usepackage{amsthm}       

\newtheorem{theorem}{Theorem}      
\newtheorem{lemma}{Lemma}     
\newtheorem{definition}{Definition} 
\newtheorem{corollary}{Corollary}  
 
\newcommand{\f}{\frac}

\begin{document}


\date{May 12, 2015}

\title{Effective integration of ultra-elliptic solutions of the focusing nonlinear Schr\"odinger equation}
\author{O. C. Wright, III\\ 
\it Department of Science and Mathematics, Cedarville University,\\ \it  251 N. Main St., Cedarville, OH 45314 
}

\maketitle

\begin{abstract}
An effective integration method based on the classical solution to the Jacobi inversion problem, using Kleinian ultra-elliptic functions, is presented for quasi-periodic two-phase solutions of the focusing nonlinear Schr\"odinger equation.  
The two-phase solutions with real quasi-periods  are known to form a two-dimensional torus, modulo a circle of complex phase factors, that can be expressed as a ratio of theta functions. In this paper, the two-phase solutions are explicitly parametrized in terms of the branch points on the genus-two Riemann surface of the theta functions.  Simple formulas,  in terms of the imaginary parts of the branch points, are obtained for the maximum modulus and the minimum modulus of the two-phase solution.
\end{abstract}
           
\vspace{.1in}



\section{Introduction}
\label{introduction}

The focusing nonlinear Schr\"odinger (NLS) equation is
\begin{equation}
i \, p_t + p_{xx} + 2  |p|^2 p =  0,
\label{nls}
\end{equation}
where $p(x,t)$ is a  complex field exhibiting focusing, viz., modulationally unstable, behavior.  The NLS equation is a well-known and thoroughly studied soliton equation applicable to a wide variety of problems, in which there is a simple balance between dispersive and nonlinear effects, ranging from shallow water waves to optical communication systems, see~\cite{kamc 00} and the references therein. The integration of multi-phase quasi-periodic solutions of the NLS equation by algebro-geometric means has also been studied extensively~\cite{dubr 81, trac 84, prev 85, trac 88, belo 94}.  A crucial feature of the algebro-geometric technique is the determination of reality conditions which must be satisfied by the integrals of motion and the invariant spectral curve, in order for the N-phase quasi-periodic solution with real quasi-periods to satisfy the scalar NLS equation, not merely a complexified coupled version of the NLS equation.  Equivalently, the constants of integration in the algebro-geometric N-phase solution
must satisfy  reality conditions, so that the initial conditions of the solution are consistent with the integrals of motion used to construct the solution. The complete characterization of all reality conditions on the constants of integration in the multi-phase solution is a crucial feature in the correct determination of the solutions.   Reality conditions for quasi-periodic solutions have been studied for the NLS equation~\cite{prev 85}  and for the sine-Gordon equation~\cite{fore 82, erco 85}, including very concrete formulas in the case of the sine-Gordon equation.  

The detailed study of two-phase solutions of the focusing  NLS equation~(\ref{nls}) is important for understanding the 
small-dispersion limit near the gradient catastrophe~\cite{bert 13} and, in particular, the generation of rogue waves near the gradient catastrophe~\cite{el 15}.  
Special classes of ultra-elliptic solutions of vector NLS equations have also been of much interest recently~\cite{poru 99, chri 00, eilb 00, wood 07}, 
including their modulation equations~\cite{kamc 14}.  For elliptic solutions of the NLS equation, Kamchatnov~\cite{kamc 00, kamc 90} has demonstrated an effective integration method that clarifies the dependence of the solution on the branch points of the underlying elliptic curve.  The purpose of this paper is to extend the effective integration method of  elliptic solutions  to two-phase quasi-periodic solutions, viz., the ultra-elliptic solutions.   As is known~\cite{prev 85, belo 94}, the real quasi-periodic solutions form a single  two-dimensional torus in the Jacobi variety of the invariant spectral curve, modulo a circle of phase factors.  In this paper, the two-phase solutions are integrated explicitly using classical methods to solve the Jacobi inversion problem in terms of Kleinian ultra-elliptic functions~\cite{bake 07, eilb 03} and, consequently, in terms of Riemann theta functions explicitly parametrized by the branch points of the genus-two Riemann surface.   The resulting formulas  are similar to  previously known algebro-geometric formulas for the solutions~\cite{belo 94}, except that the dependence of the constants of integration on the branch points is more explicit.  In particular, it is shown that, for all smooth two-phase solutions, there is a simple dependence of the maximum modulus and the minimum modulus of the two-phase solution on the imaginary parts of the branch points, see Theorems~\ref{stheorem} and~\ref{finaltheorem}.

\section{Lax Pair and Recursion Relations}
The integrability of the NLS equation~(\ref{nls}) is established through the equivalence of system~(\ref{nls}) and the commutation of a Lax pair of linear eigenvalue problems,
\begin{equation}
\begin{array}{cccccc}
\psi_x & = & \mathbb{U} \psi, & \psi_t & = & \mathbb{V} \psi,
\end{array}
\label{laxpair}
\end{equation}
where
\begin{equation}
\begin{array}{rcl}
\mathbb{U} & = & \left(\begin{array}{cc} - i \lambda & i p \\ 
			                                 i  p^* & i \lambda  
			                                 \end{array} \right), \\[.3in] 
\mathbb{V} & = & \left(\begin{array}{cc} - 2 i \lambda^2 + i  |p|^2  & 2 i \lambda p - p_x \\
                                        2 i   \lambda p^* +  p^*_x & 2 i \lambda^2 - i  |p|^2  
                                        \end{array} \right),
\end{array}
\end{equation}
$p^*$ denotes the complex conjugate of $p,$ and $\lambda$ is the spectral parameter in the inverse spectral theory of the integrable system.

The commutation of the Lax pair of linear operators~(\ref{laxpair}) implies the NLS equation in zero-curvature form
\begin{equation}
\mathbb{U}_t - \mathbb{V}_x + [\mathbb{U},\mathbb{V}] = 0,
\label{zero}
\end{equation}
where $[\mathbb{A},\mathbb{B}] = AB-BA$ is the usual operation of matrix commutation.     
The stationary N-phase solutions of the NLS equation are, by definition, solutions of the stationary NLS hierarchy defined by
\begin{equation}
\Psi_x = [\mathbb{U}, \Psi],
\label{stationaryx}
\end{equation}
in which the solution matrix $\Psi$ is polynomial in the spectral parameter $\lambda.$  The time-dependent N-phase solutions are then obtained by explicitly constructing the compatible time-dependence of $\Psi$ such that
\begin{equation}
\Psi_t = [\mathbb{V}, \Psi],
\label{stationaryt}
\end{equation}
which, in turn, implies the zero-curvature representation of the NLS equation~(\ref{zero}). 

The commutation operators in equations~(\ref{stationaryx}) and~\ref{stationaryt}) imply that the trace of $\Psi$ is constant and so, without loss of generality, constant multiples of the identity matrix may be added to $\Psi,$ and we may assume that the trace of $\Psi$ is zero.  Furthermore, the commutation structure also implies that the characteristic polynomial of $\Psi$ is a constant, providing integrals of motion that enable the integration of the N-phase solutions.

Therefore a solution to equation~(\ref{stationaryx}) is constructed of the form,
\begin{equation}
\begin{array}{rcl}
\Psi & = & \left(\begin{array}{ccc} \Psi_{11}& \Psi_{12}\\
					\Psi_{21} & -\Psi_{11}
			     \end{array} \right),
\end{array}
\label{Psi}
\end{equation}
so that the Lax pair of equations~(\ref{stationaryx}) and~(\ref{stationaryt}) become
\begin{equation}
\begin{array}{rcl}
\Psi_{11x} & =& -i  p^* \Psi_{12} + i p \Psi_{21}, \\
\Psi_{12x} & =& - 2 i p \Psi_{11} - 2 i \lambda \Psi_{12}, \\
\Psi_{21x} & = & 2 i  p^* \Psi_{11} + 2 i \lambda \Psi_{21},\\
\Psi_{11t} & = & - (2 i  \lambda p^* +  p^*_x) \Psi_{12} + (2 i \lambda p - p_x) \Psi_{21}, \\
\Psi_{12t} &= & -2(2 i \lambda p - p_x) \Psi_{11} + 2 (-2 i \lambda^2 + i  |p|^2) \Psi_{12}, \\
\Psi_{21t} & = & 2 (2 i \lambda  p^* +  p^*_x) \Psi_{11} + 2 ( 2 i \lambda^2 - i |p|^2) \Psi_{21}.
\end{array}
\label{psiequations}
\end{equation}

N-phase solutions correspond to $\Psi$ which are polynomial in $\lambda.$ Substitution of the series
\begin{equation}
\begin{array}{rcl}
\Psi_{11} & = & 1 + \sum\limits^\infty_{n = 1} f_{n} \lambda^{-n},\\[.1in]
\Psi_{12} & = & \sum\limits^\infty_{n=1} g_{n} \lambda^{-n},\\[.1in]
\Psi_{21} & = & \sum\limits^\infty_{n=1} h_{n} \lambda^{-n},
\end{array}
\end{equation}
 into the stationary equation~(\ref{stationaryx}),
produces the Lenard recursion relations for the entries of $\Psi,$
\begin{equation}
\begin{array}{rcl}
g_{n} & = & \f{1}{2} i \f{\partial g_{n-1}}{\partial x} - p f_{n-1}, \\[.1in]
h_{n} & = & -\f{1}{2} i \f{\partial h_{n-1}}{\partial x} - p^* f_{n-1},\\[.1in]
f_{n} & = &  i \int p h_{n} -  p^* g_{n} \,dx.
\end{array}
\end{equation} 
It can be shown that the integrands are always exact~\cite{flas 83}, so that the entries of $\Psi$ are differential polynomials in $p$ and $p^*.$
The homogeneous solution, for which all the constants of integration are set equal to zero, is, up to $O(\lambda^{-3}),$
\begin{equation}
\begin{array}{rclrclrcl}
f_1 & = & 0,&  g_1& = & -p, &  h_1 & = & -  p^*,\\[.1in]
f_2 & = & -\f{1}{2}  |p|^2, & g_2 & = & -\f{1}{2} i \f{\partial p}{\partial x},&  h_2 & = & \f{1}{2} i  \f{\partial p^*}{\partial x},\\[.1in]
f_3 & = & \f{1}{4} i  \left( p \f{\partial p^*}{\partial x} - p^* \f{\partial p}{\partial x} \right),& g_3 & = &\f{1}{4} \f{\partial^2 p}{\partial x^2} + \f{1}{2}  p |p|^2, &h_3 & = & \f{1}{4}  \f{\partial^2 p^*}{\partial x^2} + \f{1}{2} p^* |p|^2.
\end{array}
\end{equation}

\section{Two-phase solutions}
Two-phase solutions correspond to a third-degree polynomial solution for $\Psi,$ the most general such solution is
\begin{equation}
\begin{array}{rcl}
\Psi_{11} & = &- i \lambda^3 - i c_2 \lambda^2 + (\f{1}{2} i |p|^2 - i c_1) \lambda + \f{1}{4}  (p p^{*\prime} - p^\prime p^*)+\f{1}{2}  i c_2 |p|^2 - i c_0, \\[,1in]
 \Psi_{12} & = & i p \lambda^2 +(-\f{1}{2} p^\prime + i c_2 p) \lambda - \f{1}{4} i p^{\prime \prime}-\f{1}{2} i  p |p|^2 -\f{1}{2} c_2 p^\prime + i c_1 p, \\[.1in]
\Psi_{21} & = &i  p^* \lambda^2 +(\f{1}{2}  p^{*\prime} +i c_2  p^*) \lambda -\f{1}{4} i  p^{*\prime \prime}-\f{1}{2} i p^* |p|^2 + \f{1}{2} c_2  p^{*\prime} + i  c_1  p^*,
\end{array}
\label{psi}
\end{equation}
where $c_0, c_1, c_2 \in \mathbb{R}$ are constants of integration.
The stationary equation for this solution is
\begin{equation}
p^{\prime \prime \prime} = -6 |p|^2 p^\prime + 2i c_2 (p^{\prime \prime} + 2  p |p|^2)+ 4 c_1 p^\prime - 8 i c_0 p.
\end{equation}
The solution $\Psi$ of equation~(\ref{psi}) has real symmetry which must be satisfied in order to obtain two-phase solutions of the scalar NLS equation~(\ref{nls}), instead of solutions to a complexified pair of coupled NLS equations.
\begin{theorem}[Reality Condition]
\begin{equation}
\Psi_{12} (\lambda) = -  (\Psi_{21}(\lambda^*))^*,
\end{equation}
hence the two roots $\mu_1$ and $\mu_2$ of $\Psi_{21} (\lambda)=0$ are the complex conjugates of the two roots of $\Psi_{12}(\lambda)=0.$  In particular,
\begin{equation}
\begin{array}{rcl}
\Psi_{11} & = & - i \lambda^3 - i c_2 \lambda^2 + (\f{1}{2} i  \nu_1 - i c_1) \lambda +\f{1}{4} i  \nu_2 +\f{1}{2}  i c_2 \nu_1 - i c_0,\\[.1in]
\Psi_{12} & = & i p (\lambda - \mu_1) (\lambda - \mu_2), \\[.1in]
\Psi_{21} & = & i  p^* (\lambda - \mu^*_1) (\lambda -\mu^*_2),
\end{array}
\label{mupolynomialpsi}
\end{equation}
where $\nu_1, \nu_2 \in \mathbb{R}$ are real variables,
\begin{equation}
\begin{array}{rcl}
\nu_1 & = & |p|^2  \geq 0, \\[.1in]
\nu_2 & = &i (p^* p^\prime - p p^{*\prime}).
\end{array}
\end{equation}
\label{reality1}
\end{theorem}
The two roots $\mu_1, \mu_2 \in \mathbb{C}$ of the equation $\Psi_{12} = 0$ are analogous to the Dirichlet eigenvalues of the KdV spectral problem and, hence, are referred to as Dirichlet eigenvalues in this context.  The solution $p(x,t)$ of the NLS equation~(\ref{nls}) is recovered from the Dirichlet eigenvalues by the trace formulas.
\begin{lemma}[Trace Formulas]
The Dirichlet eigenvalues $\mu_1$ and $\mu_2$ satisfy the trace formulas,
\begin{equation}
\begin{array}{rcl}
\mu_1+\mu_2 & = & -\f{1}{2} i \f{p^\prime}{p} - c_2,\\[.1in]
\mu_1 \mu_2 & = & -\f{1}{4} \f{p^{\prime \prime}}{p} -\f{1}{2}  |p|^2 + \f{1}{2} i c_2 \f{p^\prime}{p} + c_1.
\end{array}
\label{trace}
\end{equation}
Also
\begin{equation}
\nu_2 = -2 \nu_1 (\mu_1+\mu_2+\mu_1^*+\mu_2^*+ 2 c_2),
\label{nu2equation}
\end{equation}
and
\begin{equation}
\begin{array}{rcl}
p^\prime & = & 2 i p (\mu_1 +\mu_2 + c_2),\\[.1in]
p^{\prime \prime} & = & -2p (2 \mu_1 \mu_2 +  \nu_1+ 2 c_2 \mu_1 +2 c_2 \mu_2 + 2 c_2^2- 2 c_1).
\end{array}
\label{traceinverse}
\end{equation}
\end{lemma}
The  evolution of the Dirichlet eigenvalues  is governed by the Dubrovin equations obtained by substitution
of equation~(\ref{mupolynomialpsi}) into the second and fifth equations of~(\ref{psiequations}), differentiating and substituting $\lambda = \mu_1$ or $\lambda=\mu_2.$
\begin{lemma}[Dubrovin Equations]
The Dirichlet eigenvalues $\mu_1$ and $\mu_2$ satisfy the system of Dubrovin equations, 
\begin{equation}
\begin{array}{rcl}
\f{\partial \mu_1}{\partial x} & = & 2 \f{\Psi_{11} (\mu_1)}{\mu_1-\mu_2}, \\[.1in]
\f{\partial \mu_2}{\partial x} & = & 2 \f{\Psi_{11} (\mu_2)}{\mu_2 - \mu_1}, \\[.1in]
\f{\partial \mu_1}{\partial t} & = & -2 (\mu_2 + c_2) \f{\partial \mu_1}{\partial x}, \\[.1in]
\f{\partial \mu_2}{\partial t} & = & -2 (\mu_1 + c_2) \f{\partial \mu_2}{\partial x}.
\end{array}
\label{dubrovinequations}
\end{equation}
\end{lemma}
The Dirichlet eigenvalues lie on trajectories in the complex plane determined by the  Dubrovin differential equations~(\ref{dubrovinequations}) and the initial conditions. However, not all
initial conditions satisfy the reality condition.  In other words, not all initial conditions are consistent with the integrals of motion of the Dubrovin equations and the assumption that the zeros of $\Psi_{12}$ and $\Psi_{21}$ are complex-conjugates of each other. The allowed initial conditions, viz., the allowed trajectories, are not known a priori (unlike the KdV case in which the Dirichlet eigenvalues must lie on certain intervals of the real line determined by the spectral problem) and must be determined in order to construct two-phase solutions from  the  Dubrovin equations.  

The effective integration method of Kamchatnov~\cite{kamc 00, kamc 90} uses the fact that the Dubrovin equations  show that the $\mu-$trajectories of the Dirichlet eigenvalues form a real two-dimensional manifold parametrized by
the two real variables $\nu_1$ and $\nu_2.$ In fact, by finding the algebraic dependence of the $\mu-$trajectories on
$\nu_1$ and $\nu_2,$  the reality conditions can be satisfied explicitly.
Moreover, the differential equations satisfied by $\nu_1$ and $\nu_2$ come from the $x-$flow and $t-$flow of the Lax pair for $\Psi$ by substitution of equations~(\ref{mupolynomialpsi}) into equations~(\ref{psiequations}).
\begin{lemma}
The real variables $\nu_1$ and $\nu_2$ satisfy the following system of equations,
\begin{equation}
\begin{array}{rcl}
\f{\partial \nu_1}{\partial x} & = & 2 i \nu_1 (\mu_1 + \mu_2 - \mu_1^* -\mu_2^*), \\[.1in]
\f{\partial \nu_2}{\partial x} & = &4 i \nu_1 (\mu^*_1 \mu^*_2 - \mu_1 \mu_2) - 2 c_2 \f{\partial \nu_1}{\partial x}, \\[.1in]
\f{\partial \nu_1}{\partial t} & = & \f{\partial \nu_2}{\partial x},\\[.1in]
\f{\partial \nu_2}{\partial t} & = & 8 i \nu_1 ( (\mu_1 -\mu^*_1) |\mu_2|^2 + (\mu_2 - \mu^*_2) |\mu_1|^2) -4 c_2 \f{\partial \nu_2}{\partial x} - 4 c_2^2 \f{\partial \nu_1}{\partial x}.
\end{array}
\label{nuequations}
\end{equation}
\end{lemma}
\begin{lemma} If $p(x,t)$ is a smooth bounded quasi-periodic solution of the NLS equation~(\ref{nls}), then $\nu_1 (x,t) = |p(x,t)|^2$ must have relative extrema as a function of $x$ and $t.$  Critical points of $\nu_1 (x,t)$ occur at distinguished values of the Dirichlet eigenvalues. In particular, equations~(\ref{nuequations}) imply that
\begin{equation}
\f{\partial \nu_1}{\partial x} = \f{\partial \nu_1}{\partial t} = 0 \Leftrightarrow \mbox{(i)}\,\mu_1=\mu_1^*, \,\mu_2= \mu_2^* \,\,\mbox{or}\,\mbox{(ii)}\, \mu_2^* = \mu_1.
\label{muc}
\end{equation}
\label{criticalpoints} 
\end{lemma}

Notice that relative extrema of $\nu_1$ are the possible bounds for the oscillations of a smooth quasi-periodic two-phase solution.  These relative extrema can only occur at the distinguished values of the Dirichlet eigenvalues.
\begin{theorem}
 If $p(x,t)$ is a smooth bounded two-phase solution of the NLS equation~(\ref{nls}), then $\nu_1 (x,t) = |p(x,t)|^2$ must oscillate on an interval of non-negative values whose endpoints are relative extrema of $\nu_1 (x,t).$  If $\nu_1 >0$ and $\mu_1 \neq \mu_2,$   then $\nu_1$ has a relative extremum  only if $\mu_1 = \mu_1^*$ and $\mu_2 = \mu_2^*.$
\label{nuextremalemma}
\end{theorem}
\begin{proof}
Explicit calculation using the Dubrovin equations~(\ref{dubrovinequations}) and equations~(\ref{nuequations}) shows that
if $\mu_1 \neq \mu_2$ and (i) $\mu_1 = \mu_1^*$ and $\mu_2=\mu_2^*$ or (ii) $\mu_1 = \mu_2^*,$  then
\begin{equation}
\f{\partial^2 \nu_1}{\partial x^2} \f{\partial^2 \nu_1}{\partial t^2} - (\f{\partial^2 \nu_1}{\partial x \partial t})^2 = 256 \nu_1^2 \Psi_{11}(\mu_1) \Psi_{11}(\mu_2) \in \mathbb{R}.
\end{equation}
However, if $\mu_1 \neq \mu_2$ and $\mu_1=\mu_2^*$ and $\nu_1>0,$ then
\begin{equation}
\f{\partial^2 \nu_1}{\partial x^2} \f{\partial^2 \nu_1}{\partial t^2} - (\f{\partial^2 \nu_1}{\partial x \partial t})^2 = -256 \nu_1^2 |\Psi_{11} (\mu_1)|^2<0,
\end{equation}
viz., there is a saddle point, instead of a relative extremum, at the critical point.
\end{proof}
It is worth noting that in order to obtain a similar result for $N-$phase solutions with $N \geq 3,$ the solution $p$ of the NLS equation must be considered as a simultaneous solution of $N$ evolutionary flows in the integrable NLS hierarchy of equations, because two conditions on the critical points of the $x-$ and $t-$flows of the solution
are insufficient to imply the condition of equation~(\ref{muc}) on more than two Dirichlet eigenvalues.

\section{Invariant Characteristic Equation}
The characteristic equation of $\Psi$ is an invariant of both the $x-$flow and the $t-$flow of the solutions.  As such, it provides the integrals of motion necessary to 
integrate the differential equations.  In particular, for the two-phase solutions, we seek the explicit parametrization of the $\mu_1$ and $\mu_2$ trajectories.
The characteristic equation of $\Psi,$ viz., the invariant spectral curve,  as defined by equations~(\ref{mupolynomialpsi}), is
\begin{equation}
\det (i w I - \Psi) = -w^2 + \mathscr{R}(\lambda) = 0,
\label{chareqn}
\end{equation}
where
\begin{equation}
\mathscr{R}(\lambda) = -\Psi_{11}^2 - \Psi_{12} \Psi_{21} = \prod\limits^6_{i=1} (\lambda-\lambda_i),
\label{Reqn}
\end{equation}
and it  naturally defines a hyperelliptic Riemann surface of arithmetic genus two with points $P=(\lambda, w(\lambda)) \in \mathbb{C}^2$ lying on a complex algebraic curve $\mathscr{K}_2$ defined by  equation~(\ref{chareqn}).  The curve $\mathscr{K}_2$ is a branched two-sheeted covering of the
Riemann sphere with two points over the point at infinity; it is assumed that the curve is nonsingular, i.e., the branch points $\lambda_i,$ for $i=1, \ldots, 6,$  are distinct.   Let the point $\infty^+$ be defined as the point over infinity with $\lambda = \infty$ and  $w(\lambda) =  \lambda^3 + O(\lambda^2),$ similarly $\infty^-$ is the point over $\lambda = \infty$ with $w(\lambda) = - \lambda^3+ O(\lambda^2).$
The curve $\mathscr{K}_2$ admits the usual hyperelliptic involution corresponding to sheet interchange,
\begin{equation}
\iota : \mathscr{K}_2 \rightarrow \mathscr{K}_2, \hspace{.2in} \iota (\lambda, w(\lambda)) = (\lambda, -w(\lambda))
\end{equation}
and, because the coefficients of $\mathscr{R}(\lambda)$ are real, an anti-holomorphic involution
\begin{equation}
*: \mathscr{K}_2 \rightarrow \mathscr{K}_2, \hspace{.2in} *(\lambda,w(\lambda)) = (\lambda^*, w(\lambda)^*),
\label{star}
\end{equation}
where we use the same symbol $*$ for the involution acting on points of $\mathscr{K}_2,$ as well as for complex conjugation of complex numbers. Note that 
the anti-holomorphic involution~(\ref{star}) leaves the sheets of the covering of the Riemann sphere unchanged, as can be seen by considering the action of $*$ on points in the vicinity of $P^+_\infty,$ viz., if $\lambda \in \mathbb{R},$ then
\begin{equation}
*(\lambda, w(\lambda)) = (\lambda^*, w(\lambda)^*) = (\lambda, w(\lambda)).
\end{equation}
The symmetry of the curve $\mathscr{K}_2$ expressed in the existence of the anti-holomorphic involution $*$ places reality conditions on the branch points of $\mathscr{K}_2,$
the integrals of motion of the $x-$ and $t-$flows.

\begin{corollary}[Real Curve]
$\mathscr{K}_2$ must be a real algebraic curve, viz., the branch points are either real or come in complex-conjugate pairs. 
\label{realcorollary2}
\end{corollary}

The Dubrovin equations~(\ref{dubrovinequations}) can now be integrated~\cite{dubr 81} on the hyperelliptic Riemann surface.
\begin{lemma}[Dubrovin equations reprised]
 Equation~(\ref{Reqn}) implies 
\begin{equation}
\Psi_{11}(\mu_j) = i \sqrt{\mathscr{R}(\mu_j)},
\label{psi11eqn}
\end{equation}
for $j=1,2.$
Therefore, the motion of the Dirichlet eigenvalues is defined on $\mathscr{K}_2$ by the Dubrovin equations~(\ref{dubrovinequations}),
\begin{equation}
\begin{array}{rcl}
\f{\partial \mu_1}{\partial x} & = & 2i \f{\sqrt{\mathscr{R} (\mu_1)}}{\mu_1-\mu_2}, \\[.1in]
\f{\partial \mu_2}{\partial x} & = & 2i \f{\sqrt{\mathscr{R} (\mu_2)}}{\mu_2 - \mu_1}, \\[,1in]
\f{\partial \mu_1}{\partial t} & = & -2 (\mu_2 + c_2) \f{\partial \mu_1}{\partial x}, \\[.1in]
\f{\partial \mu_2}{\partial t} & = & -2 (\mu_1 + c_2) \f{\partial \mu_2}{\partial x}.
\end{array}
\label{dubrovinequations2}
\end{equation}
\end{lemma}
\begin{lemma}
The $x-$flow and $t-$flow of $\mu_1$ and $\mu_2$ on $\mathscr{K}_2$ are linearized by the  Abel mapping,
\begin{equation}
\begin{array}{rcl}
 \int\limits^{\mu_1}_{\mu_{10}} \f{d \mu_1}{\sqrt{\mathscr{R}(\mu_1)}} + \int\limits^{\mu_2}_{\mu_{20}} \f{d \mu_2}{\sqrt{\mathscr{R}(\mu_2)}} &=& 4 i t,\\[.2in]
 \int\limits^{\mu_1}_{\mu_{10}} \f{\mu_1 d \mu_1}{\sqrt{\mathscr{R}(\mu_1)}}+\int\limits^{\mu_2}_{\mu_{20}} \f{\mu_2 d\mu_2}{\sqrt{\mathscr{R}(\mu_2)}}& = & 2 i x-4 ic_2 t,
\end{array}
\label{linearized}
\end{equation}
where $\mu_{10}$ and $\mu_{20}$ are constants of integration, viz., the initial values for $\mu_1$ and $\mu_2.$  
\end{lemma}
It is important to remember that the values of $\mu_{10}$ and $\mu_{20}$ are not arbitrary but must satisfy the reality conditions of the integrals of motion and, hence, on the loci of $\mu_1$ and $\mu_2$  by the characteristic equation~(\ref{chareqn}).  The correct determination of these initial conditions is one of the main goals of this paper.

\section{Integrals of motion and the Dirichlet eigenvalues}
The symmetric polynomials of degree $i,$ $\Sigma_i,$ $i = 1, \ldots, 4,$ of $\mu_1, \mu_2, \mu^*_1, \mu^*_2,$ must satisfy the following reality condition because they must come in complex-conjugate pairs.
\begin{corollary}[Symmetric Polynomials of the Dirichlet Eigenvalues]
\begin{equation}
\Sigma_1, \Sigma_2, \Sigma_3, \Sigma_4 \in \mathbb{R}
\end{equation} and 
\begin{equation}
\Sigma_4 = |\mu_1 \mu_2|^2 \geq 0.
\end{equation}
\label{realcorollary3}
\end{corollary}
Moreover, the algebraic constraints on the real loci of $\mu_{1}$ and $\mu_{2}$ can be made explicit.  Using the invariant
spectral polynomial~(\ref{chareqn}) and definition~(\ref{Reqn}), the symmetric polynomials   of degree $i,$ $\Lambda_i,$ $ i = 1, \ldots, 6,$ of $\lambda_i,$ $i=1, \ldots 6,$ 
the branch points of the invariant spectral curve, can be written in terms of the symmetric polynomials of degree $i,$ $\Sigma_i,$ $i = 1, \ldots, 4,$ of the Dirichlet eigenvalues $\mu_1, \mu_2, \mu^*_1, \mu^*_2.$  
\begin{lemma}[Symmetric Polynomials of Branch Points] 
\begin{equation}
\begin{array}{rcl}
\Lambda_1 & = & -2 c_2, \\[.1in]
\Lambda_2 &  = & 2 c_1 + c_2^2, \\[.1in]
\Lambda_3 & = & -2c_1 c_2 - 2 c_0 + \f{1}{2}   \nu_2 + 2 c_2   \nu_1 +   \nu_1  \Sigma_1,\\[.1in]
\Lambda_4 & = & -\f{1}{2} c_2 \nu_2 -  c_2^2 \nu_1 -  c_1 \nu_1 +\f{1}{4} \nu_1^2 + 2 c_0 c_2 + c_1^2 +  \nu_1 \Sigma_2,\\[.1in]
\Lambda_5 & = & -\f{1}{4} \nu_1 \nu_2 -\f{1}{2} c_2 \nu_1^2 +  c_0 \nu_1 +\f{1}{2}  c_1 \nu_2 - 2 c_0 c_1 +c_1 c_2  \nu_1 + \nu_1 \Sigma_3,\\[.1in]
\Lambda_6 & = & \f{1}{4} c_2 \nu_1 \nu_2 - c_2 c_0  \nu_1-\f{1}{2}  c_0 \nu_2 +\f{1}{4} c_2^2 \nu_1^2 + c_0^2 + \f{1}{16} \nu_2^2 +  \nu_1 \Sigma_4. 
\end{array}
\label{lambdaequations}
\end{equation}
\end{lemma}
Equation~(\ref{nu2equation}) and the first three equations of~(\ref{lambdaequations}) determine the three constants $c_0, c_1, c_2,$ in terms of the branch points.
\begin{lemma}[Constants of Integration]
\begin{equation}
\begin{array}{rcl}
c_2 & = & -\f{1}{2} \Lambda_1, \\[.1in]
c_1 & = & \f{1}{2} \Lambda_2-\f{1}{8} \Lambda_1^2,\\[.1in]
c_0 & = & -\f{1}{2} \Lambda_3 + \f{1}{4} \Lambda_2 \Lambda_1 - \f{1}{16} \Lambda_1^3.
\end{array}
\end{equation}
\end{lemma}
Equation~(\ref{nu2equation}) and the last three equations of~(\ref{lambdaequations}) show that loci of the Dirichlet eigenvalues can be parametrized in terms of $\nu_1$ and $\nu_2$ by solving for $\Sigma_i,$ $ i = 1, \ldots, 4.$
\begin{lemma}[Symmetric Polynomials of Dirichlet Eigenvalues]  
\begin{equation}
\begin{array}{rcl}
\Sigma_1 & = & -\f{1}{2} \f{\nu_2}{\nu_1}+\Lambda_1,\\[.1in]
\Sigma_2 & = & -\f{1}{4} \Lambda_1 \f{\nu_2}{\nu_1} -\f{1}{4} \nu_1+ \f{-\f{5}{64} \Lambda_1^4 +\f{3}{8} \Lambda_1^2 \Lambda_2 -\f{1}{4} \Lambda_2^2 -\f{1}{2}\Lambda_1\Lambda_3+ \Lambda_4}{\nu_1}+\f{1}{2}\Lambda_2 + \f{1}{8} \Lambda_1^2,\\[.1in]
\Sigma_3 & = & \f{1}{4}  \nu_2+(\f{1}{16} \Lambda_1^2-\f{1}{4}\Lambda_2) \f{\nu_2}{\nu_1} -\f{1}{4} \Lambda_1 \nu_1\\
&& +\f{\f{1}{64} \Lambda_1^5-\f{1}{8}\Lambda_1^3 \Lambda_2+\f{1}{8} \Lambda_1^2 \Lambda_3+\f{1}{4} \Lambda_1 \Lambda_2^2-\f{1}{2}\Lambda_2\Lambda_3+\Lambda_5}{\nu_1} +\f{1}{2}\Lambda_3,\\[.1in]
\Sigma_4 & =  & \f{1}{8} \Lambda_1 \nu_2-\f{1}{16} \f{\nu_2^2}{\nu_1} -\f{1}{16} \Lambda_1^2\nu_1
+(-\f{1}{32}\Lambda_1^3+\f{1}{8}\Lambda_1\Lambda_2-\f{1}{4}\Lambda_3)\f{\nu_2}{\nu_1}\\
&&+ \f{-\f{1}{256}\Lambda_1^6+\f{1}{32}\Lambda_1^4\Lambda_2-\f{1}{16}\Lambda_1^3\Lambda_3-\f{1}{16}\Lambda_1^2\Lambda_2^2+\f{1}{4} \Lambda_1\Lambda_2\Lambda_3 -\f{1}{4} \Lambda_3^2 + \Lambda_6}{\nu_1}\\
&&+\f{1}{32} \Lambda_1^4-\f{1}{8}\Lambda_1^2\Lambda_2+\f{1}{4} \Lambda_1 \Lambda_3.
\end{array}
\label{sigmaequations}
\end{equation}
\end{lemma}

Note that $\Lambda_i \in \mathbb{R},$ $i = 1, \ldots, 6,$  are real parameters and $ \nu_1, \nu_2 \in \mathbb{R}$ are real variables.  Thus the four roots $\mu_1, \mu_2, \mu_1^*, \mu_2^* \in \mathbb{C}$ of the quartic equation
\begin{equation}
\prod\limits^2_{i=1} (\mu - \mu_i) (\mu-\mu^*_i) = \mu^4 - \Sigma_1 \mu^3 + \Sigma_2 \mu^2 - \Sigma_3 \mu + \Sigma_4 = 0,
\label{quartic}
\end{equation}
each lie on a two-real-dimensional manifold parametrized by $\nu_1, \nu_2 \in \mathbb{R}.$  
The explicit solutions for $\mu_1$ and $\mu_2$ can be found using the standard procedure for solving quartic equations.  First the quartic polynomial
in equation~(\ref{quartic}) is reduced by the transformation 
\begin{equation}
 \mu = \hat{\mu} + \f{1}{4} \Sigma_1,
\label{mushift}
\end{equation}
 so that equation~(\ref{quartic}) becomes
\begin{equation}
 \hat{\mu}^4 + \hat{\Sigma}_2 \hat{\mu}^2 -\hat{\Sigma}_3 \hat{\mu} +\hat{\Sigma}_4 = 0,
\label{reducedquartic}
\end{equation}
where
\begin{equation}
\begin{array}{rcl}
\hat{\Sigma}_2 & = & -\f{3}{8} \Sigma_1^2 + \Sigma_2,\\[.1in]
\hat{\Sigma}_3& = &\f{1}{8} \Sigma_1^3-\f{1}{2} \Sigma_1 \Sigma_2 + \Sigma_3,\\[.1in]
\hat{\Sigma}_4 &= &-\f{3}{256} \Sigma_1^4 + \f{1}{16} \Sigma_1^2 \Sigma_2 -\f{1}{4} \Sigma_1 \Sigma_3 + \Sigma_4.
\end{array}
\end{equation}
Note that $\Sigma_1 \in \mathbb{R},$ so the reality condition of Corollary~\ref{realcorollary3} transfers unchanged to the shifted variables $\hat{\Sigma}_2, \hat{\Sigma}_3, \hat{\Sigma}_4 \in \mathbb{R},$ and $\hat{\Sigma}_4 \geq 0,$ because the roots of equation~(\ref{reducedquartic}) still come in two complex-conjugate pairs.
Completing the square of the first two terms of equation~(\ref{reducedquartic}),
\begin{equation}
(\hat{\mu}^2 + \f{1}{2} \hat{\Sigma}_2)^2 =  \hat{\Sigma}_3 \hat{\mu} +\f{1}{4} \hat{\Sigma}_2^2- \hat{\Sigma}_4 .
\label{mu2eqn}
\end{equation}

If $\hat{\Sigma}_3 =0,$ then equation~(\ref{mu2eqn}) can be solved immediately to give,
\begin{equation}
\hat{\mu} = \f{\pm1}{\sqrt{2}} \sqrt{- \hat{\Sigma}_2  + s \sqrt{\hat{\Sigma}_2^2 -4 \hat{\Sigma}_4}},
\label{muhatsolution1}
\end{equation}
where $s  = \pm 1.$
The reality condition of Theorem~\ref{reality1} implies that the shifted Dirichlet eigenvalues come in complex-conjugate pairs, so their values are constrained and the expressions for them in equation~(\ref{muhatsolution1}) can be simplified.

\begin{corollary}[First Symmetric Polynomial Constraint]
If $\hat{\Sigma}_3 = 0,$ then $\hat{\Sigma}_2 \geq - 2 \sqrt{\hat{\Sigma}_4},$ 
and the explicit expressions for the Dirichlet eigenvalues fall into two subcases.
\begin{enumerate}
\item If $\hat{\Sigma}_2^2 \leq 4\hat{\Sigma}_4,$ then 
\begin{equation}
\begin{array}{rcl}
\hat{\mu}_1 &=& -\f{1}{2} \sqrt{-\hat{\Sigma}_2 + 2 \sqrt{\hat{\Sigma}_4}}\pm\f{1}{2}i \sqrt{\hat{\Sigma}_2+ 2 \sqrt{\hat{\Sigma}_4}},\\[.1in]
\hat{\mu}_2 & = & \f{1}{2} \sqrt{-\hat{\Sigma}_2 + 2 \sqrt{\hat{\Sigma}_4}}\pm\f{1}{2}i \sqrt{\hat{\Sigma}_2+ 2 \sqrt{\hat{\Sigma}_4}}.
\end{array}
\label{muexpressions2}
\end{equation}
 \item  If $\hat{\Sigma}_2 > 2 \sqrt{\hat{\Sigma}_4},$ then  
\begin{equation}
\begin{array}{rcl}
\hat{\mu}_1 & = &  \pm \f{1}{\sqrt{2}}i \sqrt{ \hat{\Sigma}_2  + \sqrt{\hat{\Sigma}_2^2 -4 \hat{\Sigma}_4}},\\[.1in]
\hat{\mu}_2 & = &  \pm \f{1}{\sqrt{2}} i\sqrt{ \hat{\Sigma}_2  -  \sqrt{\hat{\Sigma}_2^2 -4 \hat{\Sigma}_4}}.
\end{array}
\label{muexpressions3}
\end{equation}
\end{enumerate}
\label{realcorollary4}
\end{corollary}

 If $\hat{\Sigma}_3 \neq 0,$ then a quantity $z,$ to be determined, is added to the quantity in the parentheses on the left-hand side of equation~(\ref{mu2eqn}),
\begin{equation}
(\hat{\mu}^2 + \f{1}{2} \hat{\Sigma}_2 + z)^2 = 2 z \hat{\mu}^2 +\hat{\Sigma}_3 \hat{\mu} + \f{1}{4} \hat{\Sigma}_2^2 -\hat{\Sigma}_4 +\hat{\Sigma}_2 z + z^2.
\label{completedwithz}
\end{equation}
The quantity $z$ is chosen so that the right-hand side of the equation is a perfect square, which means that $z$ should be chosen so that the discriminant
of the right-hand side, a quadratic polynomial in $\hat{\mu},$ is zero, viz., 
\begin{equation}
z^3 + \hat{\Sigma}_2 z^2 + (\f{1}{4} \hat{\Sigma}_2^2 - \hat{\Sigma}_4 )z -\f{1}{8} \hat{\Sigma}_3^2 = 0.
\label{cubiczequation}
\end{equation}
In fact the roots of equation~(\ref{cubiczequation}) can be written down explicitly, if necessary.
\begin{lemma} If $\hat{\Sigma}_3 \neq 0,$ then equation~(\ref{cubiczequation}) has at least one positive root.
The three roots of the cubic equation~(\ref{cubiczequation}) are given by the explicit formula,
\begin{equation}
z = -\f{1}{3} \hat{\Sigma}_2  + \sqrt[3]{R + \sqrt{\Delta}} -\f{Q}{\sqrt[3]{R+ \sqrt{\Delta}}},
\label{zformula}
\end{equation}
where 
\begin{equation}
\begin{array}{rcl}
Q & = & -\f{1}{36} \hat{\Sigma}_2^2 - \f{1}{3} \hat{\Sigma}_4, \\
R & = & \f{1}{216} \hat{\Sigma}_2^3-\f{1}{6} \hat{\Sigma}_2 \hat{\Sigma}_4 + \f{1}{16} \hat{\Sigma}_3^2,\\
\Delta & = & Q^3 + R^2,
\end{array}
\end{equation}
and $\Delta \leq 0$ if and only if all three roots are real.
A double root occurs if and only if $\Delta = 0.$
\end{lemma}

\begin{definition}
If $\hat{\Sigma}_3 \neq 0,$ then $z$ is chosen, without loss of generality, to be a  positive root of Equation~(\ref{cubiczequation}).
\end{definition}
Since $z$ is a root of the cubic equation~(\ref{cubiczequation}), the right-hand side of
equation~(\ref{completedwithz}) is a perfect square,
\begin{equation}
(\hat{\mu}^2 + \f{1}{2} \hat{\Sigma}_2 + z)^2 = 2 z (\hat{\mu} + \f{1}{4 z} \hat{\Sigma}_3)^2.
\end{equation}
 Therefore
\begin{equation}
\hat{\mu}^2 + s \sqrt{2 z} \hat{\mu} +\f{1}{2} \hat{\Sigma}_2 + s \f{1}{2 \sqrt{2 z}} \hat{\Sigma}_3 +z= 0,
\label{muquadratic}
\end{equation}
where $s = \pm 1.$ 
Finally, solving the quadratic equation~(\ref{muquadratic}) gives explicit expressions for the shifted Dirichlet eigenvalues, together with algebraic constraints.

\begin{corollary}[Second Symmetric Polynomial Constraint]
If $\hat{\Sigma}_3 \neq 0,$ then $z>0$ is chosen to be a root of equation~(\ref{cubiczequation}) and the reality condition on the Dirichlet eigenvalues, viz., they must form two complex-conjugate pairs, implies that  
\begin{equation}
\begin{array}{rcl}
\rho_+ & = & \hat{\Sigma}_2 +  \f{1}{\sqrt{2 z}} \hat{\Sigma}_3 + z \geq 0,\\
\rho_- & = & \hat{\Sigma}_2 -  \f{1}{\sqrt{2 z}} \hat{\Sigma}_3 + z \geq 0.
\end{array}
\end{equation}
The explicit  expressions
for $\hat{\mu}_1=\mu_1 - \f{1}{4} \Sigma_1$ and $\hat{\mu}_2 = \mu_2 - \f{1}{4} \Sigma_1$ in terms of $\nu_1$ and $\nu_2$ are  
\begin{equation}
\begin{array}{rcl}
\hat{\mu}_1 & = & -\sqrt{\f{z}{2}} \pm \f{1}{\sqrt{2}} i \sqrt{\hat{\Sigma}_2 +  \f{1}{\sqrt{2 z}} \hat{\Sigma}_3 + z},\\
\hat{\mu}_2 & = & \sqrt{\f{z}{2}} \pm \f{1}{\sqrt{2}} i \sqrt{\hat{\Sigma}_2 - \f{1}{\sqrt{2 z}} \hat{\Sigma}_3 + z}.
\end{array}
\label{muexpressions1}
\end{equation}
In particular, if $\hat{\Sigma}_3 \neq 0,$ then $\mu_1 \neq \mu_2.$
\label{realcorollary5}
\end{corollary}
When $\hat{\Sigma}_3 \neq 0,$ $z>0,$ so the set of $(\nu_1,\nu_2)$ for which $\hat{\Sigma}_3 \neq 0$ and either $\rho_+ = 0$ or $\rho_-=0$ is equivalent to the
set of solutions of
\begin{equation}
2 z (z + \hat{\Sigma}_2)^2 - \hat{\Sigma}_3^2 = 0.
\label{newcubic}
\end{equation}
Now $z>0$ is also the root of the cubic equation~(\ref{cubiczequation}).  The set of points $(\nu_1,\nu_2)$ for which Equations~(\ref{cubiczequation}) and~(\ref{newcubic}) have a common root is given by the resultant of the two equations.  The resultant of Equations~(\ref{cubiczequation}) and~(\ref{newcubic}) is 
a non-zero constant multiple of
\begin{equation}
\hat{\Sigma}_3^2 (16 \hat{\Sigma}_4 (\hat{\Sigma}_2^2-4 \hat{\Sigma}_4)^2 - \hat{\Sigma}_3^2 (4 \hat{\Sigma}_2^3-144\hat{\Sigma}_2 \hat{\Sigma}_4 + 27 \hat{\Sigma}_3^2)).
\end{equation}

\begin{theorem}[Boundary of the Real Parameter Set]
Suppose $\nu_1>0,$ so that $\hat{\Sigma}_2,\hat{\Sigma}_3,$ and $\hat{\Sigma}_4$ are all defined, and also assume that $\hat{\Sigma}_3 \neq 0.$ Then the boundary of the set of points $(\nu_1,\nu_2)$ satisfying the reality condition of Corollary~\ref{realcorollary5} is a subset of the algebraic set 
\begin{equation}
16 \hat{\Sigma}_4 (\hat{\Sigma}_2^2-4 \hat{\Sigma}_4)^2 - \hat{\Sigma}_3^2 (4 \hat{\Sigma}_2^3-144\hat{\Sigma}_2 \hat{\Sigma}_4 + 27 \hat{\Sigma}_3^2)=0,
\label{algebraicset}
\end{equation}
representing the set of $(\nu_1,\nu_2)$ for which equations~(\ref{cubiczequation}) and~(\ref{newcubic}) possess a common root $z_c.$  On this algebraic set, the common root $z_c$ of the two equations~(\ref{cubiczequation}) and~(\ref{newcubic}) is either
\begin{equation}
z_c = -\f{1}{3} \hat{\Sigma}_2 + \f{1}{3} \sqrt{\hat{\Sigma}_2^2 + 12 \hat{\Sigma}_4} > 0,
\label{rootzcpositive}
\end{equation}
or
\begin{equation}
z_ c=  -\f{1}{3} \hat{\Sigma}_2 - \f{1}{3} \sqrt{\hat{\Sigma}_2^2 + 12 \hat{\Sigma}_4} < 0.
\label{rootzcnegative}
\end{equation}
\label{boundarytheorem}
\end{theorem}

The effective integration of the $\mu_1$ and $\mu_2$ equations is accomplished by transforming to the evolution equations for $\nu_1$ and $\nu_2.$
The equations~(\ref{nuequations}) for $\nu_1$ and $\nu_2$ become, in the shifted variables $\hat{\mu}_1$ and $\hat{\mu}_2,$
\begin{equation}
\begin{array}{rcl}
\f{\partial \nu_1}{\partial x} & = & 2 i \nu_1 (\hat{\mu}_1 +\hat{\mu}_2 -\hat{\mu}_1^*  -\hat{\mu}_2^*), \\[.1in]
\f{\partial \nu_2}{\partial x} & = &4 i \nu_1 (\hat{\mu}^*_1 \hat{\mu}^*_2 - \hat{\mu}_1 \hat{\mu}_2) +(-\f{1}{2}\Sigma_1- 2 c_2) \f{\partial \nu_1}{\partial x}, \\[.1in]
\f{\partial \nu_1}{\partial t} & = & \f{\partial \nu_2}{\partial x},\\[.1in]
\f{\partial \nu_2}{\partial t} & = & 8 i \nu_1 ( (\hat{\mu}_1 -\hat{\mu}^*_1) |\hat{\mu}_2|^2 + (\hat{\mu}_2 - \hat{\mu}^*_2) |\hat{\mu}_1|^2) \\[.1in]
&&+(-\f{1}{4} \Sigma_1^2-2 c_2 \Sigma_1 - 4 c_2^2) \f{\partial \nu_1}{\partial x}- (\Sigma_1 +4 c_2) \f{\partial \nu_2}{\partial x}.
\end{array}
\end{equation}
Equivalently, the explicit expressions for $\hat{\mu}_1$ and $\hat{\mu}_2$ can be used to write down the complicated but explicit equations for the evolution
of $\nu_1$ and $\nu_2.$
\begin{equation}
\begin{array}{rcl}
\f{\partial \nu_1}{\partial x} & = & -2 \sqrt{2} \nu_1 (\sqrt{\rho_+} + \sqrt{\rho_-}),\\[.1in]
\f{\partial \nu_2}{\partial x} & = & 4 \nu_1 \sqrt{z} (\sqrt{\rho_+} - \sqrt{\rho_-})+\sqrt{2}\nu_1(\Sigma_1+4 c_2)(\sqrt{\rho_+}+\sqrt{\rho_-}),\\[.1in]
\f{\partial \nu_1}{\partial t} & = & \f{\partial \nu_2}{\partial x}, \\[.1in]
\f{\partial \nu_2}{\partial t} & = & -4 \sqrt{2} \nu_1 ( (z+\rho_-)\sqrt{\rho_+}+(z+\rho_+)\sqrt{\rho_-})\\[.1in]
&&+\sqrt{2}\nu_1(-\f{1}{2}\Sigma_1^2-4 c_2 \Sigma_1 -8 c_2^2)(\sqrt{\rho_+}+\sqrt{\rho_-})\\[.1in]
&&-4\nu_1\sqrt{z}(\Sigma_1 +4 c_2)(\sqrt{\rho_+}-\sqrt{\rho_-}).
\end{array}
\label{nurho}
\end{equation}
Assuming that a two-dimensional region of allowed initial values for $(\nu_1, \nu_2)$ exists, it must be bounded by a subset of the algebraic set defined by equation~(\ref{algebraicset}).  Moreover, assuming existence of solutions to equations~(\ref{nurho}), viz., a region in which both $\rho_+ \geq 0$ and $\rho_- \geq 0,$ the slopes of the boundary curves can be determined directly, viz.,   if $\nu_1 > 0$ and $\rho_+ \neq 0,$ then $\rho_- =0$ implies
\begin{equation}
\f{\partial \nu_2}{\partial \nu_1} = -\sqrt{2z} -\f{1}{2} \Sigma_1 - 2 c_2,
\end{equation}
in which either $x$ or $t$ is held constant.
Similarly, if $\nu_1 > 0$ and $\rho_- \neq 0,$ then $\rho_+ = 0$ implies
\begin{equation}
\f{\partial \nu_2}{\partial \nu_1} = \sqrt{2z}-\f{1}{2}\Sigma_1 - 2 c_2.
\end{equation}
Thus, for $\nu_1 > 0,$ the one-dimensional varieties defined by $\rho_+(\nu_1,\nu_2)=0$ or $\rho_-(\nu_1,\nu_2)=0$ define smooth curves $\nu_2 =\nu_2 (\nu_1)$ which cross transversely at points where $z \neq 0$ and $\rho_+ = \rho_- =0.$ At points where $z =0$ and $\rho_+=\rho_-=0$ the two curves 
$\rho_+=0$ and $\rho_-=0$ meet tangentially.  

\section{Dependence of Extrema on the Branch Points}
Having obtained the explicit expressions for $\mu_1$ and $\mu_ 2$ in terms of $\nu_1$ and $\nu_2$ in Corollary~\ref{realcorollary5}, these expressions can
be used to solve the $\nu_1$ and $\nu_2$ equations of equation~(\ref{nuequations}).  Assuming that a smooth bounded two-phase solution exists, there must be at least one positive relative extremum for $\nu_1,$ which can occur only if $\mu_1 = \mu_1^*$ and $\mu_2 = \mu_2^*,$ provided $\mu_1 \neq \mu_2.$  It is now possible to characterize explicitly all cases of the disposition of the branch points of $\mathscr{K}_2$ for which the solution for $\nu_1$ could possibly oscillate on an 
interval of non-negative values.

\begin{lemma}
$\mu_1 = \mu_1^*$ and $\mu_2=\mu_2^*$ if and only if $\hat{\Sigma}_3 = 0$ and $\hat{\Sigma}_2=- 2 \sqrt{\hat{\Sigma}_4}.$
Also, $\mu_1 = \mu_2^*$ or $\mu_1 = \mu_2$ if and only if $\hat{\Sigma}_3 = 0$ and $\hat{\Sigma}_2=2 \sqrt{\hat{\Sigma}_4}.$ 
\label{mutosigma}
\end{lemma}
\begin{proof}
The explicit expressions in equations~(\ref{muexpressions2}),~(\ref{muexpressions3}), and~(\ref{muexpressions1}) imply the stated results.
\end{proof}

\begin{lemma}
If $p(x,t)$ is a smooth bounded two-phase solution of the NLS equation~(\ref{nls}) and $\nu_1(x,t) = |p(x,t)|^2,$ then the non-zero extrema of $\nu_1(x,t)$ 
must occur at points $(\nu_1,\nu_2)$ where $\hat{\Sigma}_3 =0$ and $ \hat{\Sigma}_2^2 - 4 \hat{\Sigma}_4 = 0.$
\end{lemma}
 \begin{proof}
Theorem~\ref{nuextremalemma} shows that extrema can only occur when either (i) $\mu_1 = \mu_2$ or (ii) $\mu_1=\mu_1^*$ and $\mu_2=\mu_2^*.$  Together
with the previous lemma, the result follows.
\end{proof}

\begin{lemma}
Suppose the quantities $\nu_1, \nu_2 \in \mathbb{R}$ are replaced by  complex variables $\nu_1, \nu_2 \in \mathbb{C}$ in the rational functions $\hat{\Sigma}_2, \hat{\Sigma}_3,$ and $\hat{\Sigma}_4.$ Then $\hat{\Sigma}_3 = 0$ and $\hat{\Sigma}_2^2 - 4 \hat{\Sigma}_4 = 0$ if and only if
$P(\nu_1)=0,$ where $P(\nu_1)$ is a polynomial of degree ten in $\nu_1 \in \mathbb{C}.$
\label{realmulemma}
\end{lemma}
\begin{proof}
$P(\nu_1)=0$ upon elimination of $\nu_2$ from the two equations $\hat{\Sigma}_3=0$ and $\hat{\Sigma}_2^2 - 4\hat{\Sigma}_4 =0.$
\end{proof}
The power of the effective integration method presented in this paper  lies in the fact, stated in the following theorem, that the roots of the polynomial equation  $P(\nu_1)=0$ can be expressed explicitly in terms of the branch points of the Riemann surface $\mathscr{K}_2.$

\begin{theorem}
If $p(x,t)$ is a smooth bounded two-phase solution  of the NLS equation~(\ref{nls}), then $\nu_1 (x,t)$ oscillates on an interval bounded either by zero or by 
positive roots  of the polynomial equation $P(\nu_1)=0.$  The polynomial $P(\nu_1)$ is a resolvent polynomial for the factorization of the reduced form of $\mathscr{R} (\lambda)$ into two cubic factors.  In particular, the ten roots of $P(\nu_1)= 0,$ together with the corresponding $\nu_2$ values, are given explicitly 
in terms of the six branch points $\lambda_i,$ $i = 1, \ldots, 6,$ of $\mathscr{K}_2$ by the formulas
\begin{equation}
\begin{array}{rcl}
\nu_1 & = & -  \f{1}{4} \chi_1^2,\\[.2in]
\nu_2 & = & -\f{1}{2}  \chi_1 \chi_2,
\end{array}
\label{nu1expression}
\end{equation}
where
\begin{equation}
\begin{array}{rcl}
\chi_1 & = & \lambda_{\pi(1)} + \lambda_{\pi(2)}+\lambda_{\pi(3)}-\lambda_{\pi(4)}-\lambda_{\pi(5)}-\lambda_{\pi(6)},\\[.2in]
\chi_2 &= & \lambda_{\pi(1)}^2 + \lambda_{\pi(2)}^2+\lambda_{\pi(3)}^2-\lambda_{\pi(4)}^2-\lambda_{\pi(5)}^2-\lambda_{\pi(6)}^2,
\end{array}
\end{equation}
and $(\pi(1),\pi(2),\pi(3),\pi(4),\pi(5),\pi(6)) \in S_6,$ the permutation group of order six. 
\label{nuexplicit}
\end{theorem}
\begin{proof}
If $\mu_1 \neq \mu_2,$ then extrema of $\nu_1$ (or saddle points) occur when $\mu_1 = \mu_1^*, \mu_2 = \mu_2^*$ (or $\mu_1 = \mu_2^*$) and so (in either case, although we are only concerned with extrema) $P(\nu_1)=0.$   For each of these ten roots, the sextic $\mathscr{R}$ factors into the product of two cubics $Q_1$ and $Q_2,$
\begin{equation}
\begin{array}{rcl}
\mathscr{R}(\lambda) & = & - \Psi_{11} (\lambda)^2 +  \nu_1 (\lambda - \mu_1)^2 (\lambda -\mu_2)^2, \\[.1in]
& = & Q_1 (\lambda) Q_2 (\lambda),
\end{array}
\end{equation}
where
\begin{equation}
\begin{array}{rcl}
Q_1 & = &i \Psi_{11} (\lambda) + i \sqrt{ \nu_1} (\lambda - \mu_1)(\lambda - \mu_2),\\[.1in] 
Q_2 & = & i \Psi_{11} (\lambda) - i \sqrt{ \nu_1} (\lambda - \mu_1)(\lambda - \mu_2). 
\end{array}
\end{equation}
Notice that, in case $\mu_1=\mu_2,$ it is still true that $P(\nu_1)=0,$ and  there is still  a similar factorization of $\mathscr{R}$ into cubic factors, so the stated 
formulas for $\nu_1$ and $\nu_2$ are still true.
The zeros of each cubic factor form two non-overlapping sets, each comprised of three distinct branch points. Using Vieta's formulas to express the coefficients of the cubic factors $Q_1$ and $Q_2$ in terms of the three roots of $Q_1 = 0$ and the three roots of $Q_2 = 0,$ explicitly solving for $\nu_1$ and $\nu_2$ from the coefficients of the quadratic and constant terms of $Q_1$ and $Q_2,$ and eliminating
the terms depending on $\mu_1$ and $\mu_2,$  leads to the stated formulas for $\nu_1$ and $\nu_2.$  Notice that both $\nu_1$ and $\nu_2$  are invariant under the actions of re-ordering the first three or the last three summands, and also swapping the first three summands with the second three summands, so that the number of values of either $\nu_1$ or $\nu_2$ is  $6!/(3!3!2!)=10,$ as expected.
 In fact, explicit calculation shows that if $\nu_1 = - a_2^2,$ then the degree ten polynomial $P(\nu_1)$ is the resolvent polynomial for the factorization into two cubic factors, one of the factors  being $\lambda^3 + a_2 \lambda^2 +a_1 \lambda + a_0,$  of the reduced sextic polynomial of the sextic polynomial $\mathscr{R}(\lambda).$   The resolvent polynomial can be calculated by successive eliminations of variables from the equations obtained by equating coefficients of the expanded product of the desired factors equal to the polynomial to be factored~\cite{piez 04}.
\end{proof}

\begin{theorem} If a  smooth bounded two-phase solution of the NLS equation~(\ref{nls}) exists, then the values of $\nu_1$ and $\nu_2$ at  critical points of $\nu_1(x,t)$ are given by
\begin{equation}
\begin{array}{rcl}
\nu_1 & = & -  \f{1}{4}  \chi_1^2,\\[.2in]
\nu_2 & = & -\f{1}{2}  \chi_1 \chi_2,
\end{array}
\label{chieqn}
\end{equation}
where
\begin{equation}
\begin{array}{rcl}
\chi_1 & = & \lambda_{\pi(1)} + \lambda_{\pi(2)}+\lambda_{\pi(3)}-\lambda_{\pi(4)}-\lambda_{\pi(5)}-\lambda_{\pi(6)},\\[.2in]
\chi_2 &= & \lambda_{\pi(1)}^2 + \lambda_{\pi(2)}^2+\lambda_{\pi(3)}^2-\lambda_{\pi(4)}^2-\lambda_{\pi(5)}^2-\lambda_{\pi(6)}^2.
\end{array}
\end{equation}
Moreover, if $\mu_1 \neq \mu_2$ at a critical point, then the following expressions of the second-order partial derivatives at the critical point have a similarly explicit dependence on the branch points,
\begin{equation}
\begin{array}{rcl}
\nu_{1xx}\nu_{1tt} - \nu_{1xt}^2 & = & 256 \nu_1^2 \Psi_{11} (\mu_1) \Psi_{11} (\mu_2) \\[.1in]
 &=& -16 \chi_1 \prod\limits^{3}_{k=1} (\lambda_{\pi(k)}-\lambda_{\pi(4)}) (\lambda_{\pi(k)}-\lambda_{\pi(5)}) (\lambda_{\pi(k)}-\lambda_{\pi(6)}), 
\end{array}
\label{del}
\end{equation}
and
\begin{equation}
\begin{array}{rcl}
\nu_{1xx} & = & 8 i \nu_1 \frac{\Psi_{11} (\mu_1) - \Psi_{11} (\mu_2)}{\mu_1 - \mu_2}\\[.1in]
&=&-2 (\lambda_{\pi(1)}-\lambda_{\pi(4)}) (\lambda_{\pi(1)}-\lambda_{\pi(5)}) (\lambda_{\pi(2)}-\lambda_{\pi(4)}) (\lambda_{\pi(2)}-\lambda_{\pi(5)}) \\
&& -2 (\lambda_{\pi(1)}-\lambda_{\pi(4)}) (\lambda_{\pi(1)}-\lambda_{\pi(6)}) (\lambda_{\pi(3)}-\lambda_{\pi(4)}) (\lambda_{\pi(3)}-\lambda_{\pi(6)}) \\
&& -2 (\lambda_{\pi(2)}-\lambda_{\pi(5)}) (\lambda_{\pi(2)}-\lambda_{\pi(6)}) (\lambda_{\pi(3)}-\lambda_{\pi(5)}) (\lambda_{\pi(3)}-\lambda_{\pi(6)}) \\
&& -2 \chi_{1} (\lambda_{\pi(1)}-\lambda_{\pi(4)}) (\lambda_{\pi(2)}-\lambda_{\pi(5)}) (\lambda_{\pi(3)}-\lambda_{\pi(6)}),
\end{array}
\label{nuxx}
\end{equation}
where $(\pi(1),\pi(2),\pi(3),\pi(4),\pi(5),\pi(6)) \in S_6,$ where $S_6$ is the permutation group of order $6.$  
\end{theorem}
\begin{proof}
The first part of the theorem is just a restatement of the result of the previous theorem, viz., by Lemma~\ref{criticalpoints}, at a critical point of $\nu_1(x,t)$ either (i) $\mu_1 = \mu_1^*$ and $\mu_2 = \mu_2^*$ or (ii) $\mu_1 = \mu_2^*.$  In either case
$P(\nu_1)=0$ by Lemmas 10 and 12.   According to Theorem~\ref{nuexplicit}, $\nu_1$ and $\nu_2$ are given by the stated formulas for one of the permutations of the indices.  The  expressions for the second-order partial derivatives at the critical point, assuming $\mu_1 \neq \mu_2,$ follow by direct substitution of  (i) $\mu_1 = \mu_1^*$ and $\mu_2 = \mu_2^*$ or (ii) $\mu_1 = \mu_2^*$ into the expressions for the second-order partial derivatives, resulting in formulas in terms of $\Psi_{11}(\mu).$ Then the explicit expressions for
$\mu_1$ and $\mu_2$ in terms of $\nu_1$ and $\nu_2$ are substituted into $\Psi_{11}(\mu)$, using equations~(\ref{mushift}) and~(\ref{muexpressions2}).  Finally the explicit expressions for $\nu_1$ and $\nu_2$ in terms of the branch points, given by equation~(\ref{chieqn}), are used. 
\end{proof}

\section{Branch Point Reality Conditions}
Using Theorem~\ref{nuexplicit}, the reality condition of Corollary~\ref{realcorollary2} on the algebraic curve can be strengthened.  
\begin{lemma}
If $\lambda_1, \lambda_2, \lambda_3, \lambda_4, \lambda_5, \lambda_6 \in \mathbb{R},$ then $P(\nu_1) =0$ has no positive roots.
\end{lemma}
\begin{proof}
The result follows immediately from equation~(\ref{nu1expression}).
\end{proof}
\begin{lemma}
If $\nu_2 \in \mathbb{R},$   $\lambda_1, \lambda_2, \lambda_3, \lambda_4, \lambda_5, \lambda_6 \in \mathbb{R}$ are all distinct, and $\lambda_1 = \lambda_2^*,$ then $P(\nu_1)=0$ has no positive roots.
\end{lemma}
\begin{proof}
The proof is by contradiction.  If $P(\nu_1) = 0$ and $\nu_1 >0,$ then equation~(\ref{nu1expression}) implies, for example, that $\lambda_3 + \lambda_4 = \lambda_5 + \lambda_6$ and, also, that $\lambda_3 \lambda_4 = \lambda_5 \lambda_6,$ since $\nu_2$ is real.  However these two equalities imply that $\lambda_3$ is equal to at least one of $\lambda_5$ or $\lambda_6,$ contradicting the distinctness of the branch points, viz., the non-singularity assumption on $\mathscr{R}.$
\end{proof}
\begin{lemma}
If $\nu_2 \in \mathbb{R},$  $\lambda_1, \lambda_2, \lambda_3, \lambda_4, \lambda_5, \lambda_6 \in \mathbb{R}$ are all distinct, and $\lambda_1 = \lambda_2^*$
and $\lambda_3 = \lambda_4^*,$  then $P(\nu_1) = 0$ has no positive roots.
\end{lemma}
\begin{proof}
The proof of this lemma is similar to the previous lemma.  If $\nu_1 >0$ in equation~(\ref{nu1expression}) and the corresponding $\nu_2 \in \mathbb{R}$ in
equation~(\ref{nu1expression}), then either $\lambda_5 = \lambda_6$ or, otherwise, $\lambda_1$ equals one of $\lambda_5$ or $\lambda_6.$
\end{proof}
\begin{theorem}[Reality Conditions for Bounded Two-Phase Solutions]
If the  NLS equation~(\ref{nls}) has a smooth bounded two-phase solution, and $\mathscr{K}_2$ is non-singular, i.e., its six branch points are distinct, then the branch points form three complex-conjugate pairs.  $P(\nu_1)=0$ has
exactly four distinct non-negative roots, these are the possible non-zero bounds of the two-phase solution, and the corresponding values of $\nu_2$ given by equation~(\ref{nu1expression}) are also real.  In particular,  if $\lambda_1 = \lambda_2^* = r_1 + i s_1,$ $\lambda_3=\lambda_4^* = r_2+i s_2$ and $\lambda_5 = \lambda_6^* = r_3 +i s_3,$  with $0<s_1 \leq s_2 \leq s_3,$ then
the four real roots of $P(\nu_1)=0$ are
\begin{equation}
0\leq\nu_1^{(1)} < \nu_1^{(2)} < \nu_1^{(3)} < \nu_1^{(4)},
\end{equation}
where
\begin{equation}
\begin{array}{rcl}
\nu_1^{(1)} & = & (s_1+s_2-s_3)^2,\\[.1in]
\nu_1^{(2)} & = & (s_1-s_2+s_3)^2,\\[.1in]
\nu_1^{(3)} & = & (-s_1+s_2+s_3)^2,\\[.1in]
\nu_1^{(4)} & = & (s_1 + s_2 + s_3)^2,
\end{array}
\end{equation}
and
\begin{equation}
\begin{array}{rcl}
\nu_2^{(1)} &=& 4(s_1+s_2-s_3)(r_1 s_1+r_2 s_2-r_3 s_3),\\[.1in]
\nu_2^{(2)}& = & 4(s_1-s_2+s_3)(r_1 s_1- r_2 s_2+ r_3 s_3),\\[.1in]
\nu_2^{(3)}& = &4(-s_1+s_2+s_3)(-r_1 s_1+r_2 s_2+r_3 s_3),\\[.1in]
\nu_2^{(4)}& = & 4(s_1+s_2+s_3)(r_1 s_1+ r_2 s_2+r_3 s_3).
\end{array}
\end{equation}
\label{focusingreality}
\end{theorem}
\begin{proof}
The preceding lemmas eliminate any other possibilities for branch points satisfying the reality conditions.  Theorem 4, equation~(\ref{nu1expression}), provides the explicit expressions for the four real roots.
The explicit expressions for $\nu_1$ show that the four real values are distinct, otherwise at least one of $s_1, s_2,$ or $s_3$ must be zero, which contradicts the assumption that the branch points are distinct.  
\end{proof}
\begin{lemma}
If a two-phase solution of the NLS equation~(\ref{nls}) exists, then it is bounded, viz., there exists a positive number $M \in \mathbb{R}$ such that $\nu_1(x,t)^2 + \nu_2 (x,t)^2 < M$ for all $x$ and $t.$  
\label{boundedlemma}
\end{lemma}
\begin{proof}

First we show that $\nu_1$ is bounded.
\begin{enumerate}
\item If $\nu_1 \rightarrow \infty$ and $\nu_2 = o(\nu_1)$ and $\Lambda_1 \neq 0,$ then $\Sigma_4 \sim -\f{1}{16} \Lambda_1^2 \nu_1 <0,$ which is impossible since $\Sigma_4 \geq 0.$
\item If $\nu_1 \rightarrow \infty$ and $\nu_2 = o(\nu_1)$ and $\Lambda_1 = 0,$ then $\Sigma_4 \sim -\f{\nu_2^2}{16 \nu_1} < 0,$ also impossible.
\item If $\nu_1 \rightarrow \infty$ and $\nu_2 = O(\nu_1),$ then $\hat{\Sigma}_2 \sim -\f{1}{4} \nu_1, \hat{\Sigma}_3 \sim -\f{1}{8} \Lambda_1 \nu_1 + \f{3}{16} \nu_2,$ and $z \rightarrow 0.$  Hence
\begin{equation}
\hat{\Sigma}_2 \pm \f{\hat{\Sigma}_3}{\sqrt{2 z}} + z \sim -\f{\nu_1}{4} \,\mbox{or}\,  \pm\f{-\f{1}{8} \Lambda_1 \nu_1 + \f{3}{16} \nu_2}{\sqrt{2z}},
\end{equation}
so that, in any case, at least one of $\rho_1$ or $\rho_2$ will become negative, which contradicts the reality condition of Corollary~\ref{realcorollary5}.
\item If $\nu_1 \rightarrow \infty$ and $\f{\nu_2}{\nu_1} \rightarrow \infty,$ then $\Sigma_4 \sim -\f{\nu_2^2}{\nu_1} < 0,$ which is also impossible.
\end{enumerate}
The preceding cases are exhaustive and show that $\nu_1$ is bounded.  

Consequently, if $|\nu_2| \rightarrow \infty,$ then
$$\Sigma_4 \sim - \f{1}{16} \f{\nu_2^2}{\nu_1},$$
which contradicts $\Sigma_4 \geq 0.$  Hence $|\nu_2|$ is also bounded.
\end{proof}

\begin{lemma}
If a smooth two-phase solution of the NLS equation~(\ref{nls}) exists in a neighborhood of a point where $\nu_1(x,t) = \nu_1^{(4)},$ then  $\nu_1 = \nu_1^{(4)} >0$ is a relative maximum of $\nu_1(x,t).$  Moreover, the values of the corresponding solutions $(\mu_1,\sigma_1 \sqrt{\mathscr{R}(\mu_1)}) \in \mathbb{R}^2$ and $(\mu_2,\sigma_2 \sqrt{\mathscr{R}(\mu_2)}) \in \mathbb{R}^2$ of the Dubrovin equations on the Riemann surface $\mathscr{K}_2$ at this point are real and are given by explicit formulas,
\begin{equation}
\begin{array}{rcl}
\mu_1 & = & \f{r_1 (s_2 + s_3)+r_2 (s_1+s_3) + r_3 (s_1+s_2) - \sqrt{A(\vec{r},\vec{s})}}{2(s_1 + s_2 + s_3)},\\
\mu_2 & = & \f{r_1 (s_2 + s_3)+r_2 (s_1+s_3) + r_3 (s_1+s_2) + \sqrt{A(\vec{r},\vec{s})}}{2(s_1 + s_2 + s_3)},
\end{array}
\label{muatnu4defn}
\end{equation}
where
\begin{equation}
\begin{array}{rcl}
A(\vec{r},\vec{s})&  =&  4 s_1 s_2 s_3 (s_1 + s_2 + s_3) + s_1^2 (r_2-r_3)^2 + s_2^2 (r_1-r_3)^2\\
&& +s_3^2 (r_1-r_2)^2 +   2 s_1 s_3 (r_2-r_3)(r_2-r_1) \\
&&   + 2 s_2 s_3 (r_1-r_3)(r_1-r_2) +2 s_1 s_2 (r_3 - r_1)(r_3-r_2), 
\end{array}
\label{Adefn}
\end{equation}
and  $A(\vec{r},\vec{s}) > 0.$ The hyperelliptic irrationalities are given by
\begin{equation}
\sigma_i \sqrt{\mathscr{R} (\mu_i)} = \sigma_i |\mu_i - \lambda_1| |\mu_i - \lambda_3| |\mu_i - \lambda_5|,
\end{equation}
for $i = 1,2,$ with the hyperelliptic sheets given by  $\sigma_1 = -1$ and $\sigma_2 = 1.$
\label{lemmanu4}
\end{lemma}
\begin{proof}
 Direct substitution shows that when $(\nu_1,\nu_2) = (\nu_1^{(4)},\nu_2^{(4)}),$
\begin{equation}
\hat{\Sigma}_2  =  - \f{A(\vec{r},\vec{s})}{2 (s_1 + s_2 +s_3)^2} < 0
\end{equation}
where
\begin{equation}
\begin{array}{rcl}
A(\vec{r},\vec{s})&  =&  4 s_1 s_2 s_3 (s_1 + s_2 + s_3) + s_1^2 (r_2-r_3)^2 + s_2^2 (r_1-r_3)^2\\
&& +s_3^2 (r_1-r_2)^2 +   2 s_1 s_3 (r_2-r_3)(r_2-r_1) \\
&&   + 2 s_2 s_3 (r_1-r_3)(r_1-r_2) +2 s_1 s_2 (r_3 - r_1)(r_3-r_2)  \\[.1in]
&=&  4 s_1 s_2 s_3 (s_1 + s_2 + s_3) + (s_1 (r_2-r_3) + s_3 (r_1-r_2))^2\\
&&  +s_2^2 (r_1-r_3)^2 +   2 s_2 s_3 (r_1-r_3)(r_1-r_2) \\
&&   + 2 s_1 s_2 (r_3-r_1)(r_3-r_2)\\[.1in]
& > & 0.
\end{array}
\label{Aeqn}
\end{equation}
The positivity of $A(\vec{r},\vec{s})$ follows from the fact that the first expression in equation~(\ref{Aeqn}) for $A(\vec{r},\vec{s})$ is symmetric with respect to interchange of the indices $i=1,2,3.$  Furthermore, if $s_1, s_2, s_3 >0$ and $r_1 \leq r_2 \leq r_3,$ then the second expression for $A$ is positive.
Thus $\hat{\Sigma}_2 = - 2 \sqrt{\hat{\Sigma}_4}$ and, hence, $\mu_1 = \mu_1^*$ and $\mu_2 = \mu_2^*,$ with
\begin{equation}
\begin{array}{rcl}
\mu_1 & = & \f{r_1 (s_2 + s_3)+r_2 (s_1+s_3) + r_3 (s_1+s_2) - \sqrt{A(\vec{r},\vec{s})}}{2(s_1 + s_2 + s_3)},\\
\mu_2 & = & \f{r_1 (s_2 + s_3)+r_2 (s_1+s_3) + r_3 (s_1+s_2) + \sqrt{A(\vec{r},\vec{s})}}{2(s_1 + s_2 + s_3)}.
\end{array}
\label{muexplicitatnu4}
\end{equation}
Obviously, $\mu_1, \mu_2 \in \mathbb{R}$ and $\mu_1 \neq \mu_2,$ as expected.  Moreover, the real symmetry of the hyperelliptic curve $\mathscr{K}_2$ implies that 
\begin{equation}
\mathscr{R} (\mu_i) = |\mu_i - \lambda_1|^2 |\mu_i - \lambda_3|^2 |\mu_i -\lambda_5|^2,
\label{Rexplicitatnu4}
\end{equation}
since $\mu_i \in \mathbb{R}$ for $i = 1,2.$

The explicit expressions of Theorem 5 can  be used to obtain the result, using standard results of differential calculus.
Equation~(\ref{psi11eqn}) is also used, in order to determine the correct hyperelliptic sheets, $\sigma_1=\pm 1$ or $\sigma_2 = \pm 1,$ for the points $(\mu_1,\sigma_1 \sqrt{\mathscr{R}(\mu_1)})$ and $(\mu_2,\sigma_2 \sqrt{\mathscr{R}(\mu_2)})$ on the Riemann surface $\mathscr{K}_2.$
Equation~(\ref{del}) becomes
\begin{equation}
\begin{array}{rcl}
\f{\partial^2 \nu_1}{\partial x^2} \f{\partial^2 \nu_1}{\partial t^2} - (\f{\partial^2 \nu_1}{\partial x \partial t})^2 & = & 256 \nu_1^2 \Psi_{11}(\mu_1) \Psi_{11}(\mu_2)\\[.1in]
&= & -256 \nu_1^2 \sigma_1 \sigma_2 \sqrt{\mathscr{R}(\mu_1)} \sqrt{\mathscr{R}(\mu_2)}, \\[.1in]
&= & 256 (s_1+s_2+s_3) s_1 s_2 s_3 |\lambda_1-\lambda_3^*|^2|\lambda_1-\lambda_5^*|^2 |\lambda_3 -\lambda_5^*|^2 \\[.1in]
& > & 0.
\end{array}
\end{equation}
Thus $\sigma_1 \sigma_2 < 0 $ and $\sigma_2 = -\sigma_1.$

Equation~(\ref{nuxx}) becomes
\begin{equation}
\begin{array}{rcl}
\f{\partial^2 \nu_1}{\partial x^2} & =  & 8 i \nu_1 \f{\Psi_{11} (\mu_1) - \Psi_{11}(\mu_2)}{\mu_1 - \mu_2} \\[.1in]
& =& 8 \nu_1 \f{s_1+s_2+s_3}{\sqrt{A(\vec{r},\vec{s})}} \sigma_1 (\sqrt{\mathscr{R}(\mu_1)}+\sqrt{\mathscr{R}(\mu_2)}) \\[.1in]
&= & -32 (s_1+s_2+s_3)s_1s_2s_3 - 8 s_1 s_2 |\lambda_1 -\lambda_3^*|^2\\
&&-8s_1s_3|\lambda_1-\lambda_5^*|^2-8s_2s_3|\lambda_3-\lambda_5^*|^2\\
&<& 0.
\end{array}
\end{equation}
Hence there is a relative maximum.  Moreover, at this relative maximum $\sigma_1 = -1$ and $\sigma_2 = 1.$
\end{proof}

\begin{lemma}
If a smooth two-phase solution of the  NLS equation~(\ref{nls}) exists in a neighborhood of a point  $(x,t,\nu_1(x,t))=(x,t,\nu_1^{(i)})$ for $i=2$ or $i=3,$ then  $(x,t,\nu_1^{(i)})$ is  a saddle point of  $\nu_1= \nu_1(x,t).$  
\end{lemma}

\begin{proof} 
Consider the case where $i=3,$ the case for $i=2$ is similar. As in the previous lemma,
\begin{equation}
\hat{\Sigma}_2  =  - \f{A_1}{2 (-s_1 + s_2 +s_3)^2},
\end{equation}
where $A_1 = A(\vec{r},-s_1,s_2,s_3)$ and $A(\vec{r},\vec{s})$ is given by equation~(\ref{Aeqn}). In this case the sign of $A_1 = A(\vec{r},-s_1,s_2,s_3)$ is indeterminate, and
\begin{equation}
\begin{array}{rcl}
\mu_1 & = & \f{r_1 (s_2 + s_3)+r_2 (-s_1+s_3) + r_3 (-s_1+s_2) - \sqrt{A_1}}{2(-s_1 + s_2 + s_3)},\\
\mu_2 & = & \f{r_1 (s_2 + s_3)+r_2 (-s_1+s_3) + r_3 (-s_1+s_2) + \sqrt{A_1}}{2(-s_1 + s_2 + s_3)}.
\end{array}
\label{muexplicit}
\end{equation}
If $A_1 \neq 0,$ then $\mu_1 \neq \mu_2$ and the formula of equation~(\ref{del}) is applicable.  If $A_1 = 0,$ then $\mu_1 = \mu_2 \in \mathbb{R},$ but all the expressions for the second-order partial derivatives of $\nu_1(x,t)$ have finite limits as $\mu_2 \rightarrow \mu_1$  with $\mu_2^*=\mu_1.$  So the formula of equation~(\ref{del}) is still valid in the limit.
In particular, regardless of the value of $A_1,$
\begin{equation}
\begin{array}{rcl}
\f{\partial^2 \nu_1}{\partial x^2} \f{\partial^2 \nu_1}{\partial t^2} - (\f{\partial^2 \nu_1}{\partial x \partial t})^2 & = & 256 \nu_1^2 \Psi_{11}(\mu_1) \Psi_{11}(\mu_2)\\[.1in]
&= & -256 (-s_1+s_2+s_3) s_1 s_2 s_3 |\lambda_1-\lambda_3|^2|\lambda_1-\lambda_5|^2 |\lambda_3 -\lambda_5^*|^2 \\[.1in]
& <& 0,
\end{array}
\end{equation}
so there is a saddle point.
\end{proof}

\begin{lemma}
If a smooth two-phase solution of the NLS equation~(\ref{nls}) exists in a neighborhood of a point $(x,t,\nu_1(x,t))=(x,t,\nu_1^{(1)}),$ then  $\nu_1 = \nu_1(t,x)$ has a saddle point at $(t,x,\nu_1^{(1)})$   when $s_3<s_1+s_2,$ but  $\nu_1(t,x)$ has a relative minimum of $\nu_1 = \nu_1^{(1)} > 0$ when $s_3 > s_1 + s_2.$
\end{lemma}

\begin{proof}
As before, consider the formula given by equation~(\ref{del}).  When $\nu_1 = \nu_1^{(1)},$
\begin{equation}
\begin{array}{rcl}
\f{\partial^2 \nu_1}{\partial x^2} \f{\partial^2 \nu_1}{\partial t^2} - (\f{\partial^2 \nu_1}{\partial x \partial t})^2 & = & 256 \nu_1^2 \Psi_{11}(\mu_1) \Psi_{11}(\mu_2)\\[.1in]
&= & -256 (s_1+s_2-s_3) s_1 s_2 s_3 |\lambda_1-\lambda_3^*|^2|\lambda_1-\lambda_5|^2 |\lambda_3 -\lambda_5|^2 
\end{array}
\end{equation}
If $s_3< s_1+s_2,$ then the above expression is negative, so there is a saddle point.
If $s_3 > s_1+s_2,$ then the above expression is positive, and
equation~(\ref{nuxx}) becomes
\begin{equation}
\begin{array}{rcl}
\f{\partial^2 \nu_1}{\partial x^2} & =  & 8 i \nu_1 \f{\Psi_{11} (\mu_1) - \Psi_{11}(\mu_2)}{\mu_1 - \mu_2} \\[.1in]
&= & 32 (s_1+s_2-s_3)s_1s_2s_3 - 8 s_1 s_2 |\lambda_1 -\lambda_3^*|^2\\
&&+8s_1s_3|\lambda_1-\lambda_5|^2+8s_2s_3|\lambda_3-\lambda_5|^2\\[.1in]
&=&8(s_1+s_2)(s_3-s_1)(s_3-s_2)(s_3-s_1-s_2)-8 s_1 s_2 (r_1-r_2)^2 \\
&& +8 s_2 s_3 (r_2-r_3)^2+8s_1 s_3(r_1-r_3)^2\\[.1in]
&>&8(s_1+s_2)(s_3-s_1)(s_3-s_2)(s_3-s_1-s_2)-8 s_1 s_2 (r_1-r_2)^2 \\
&& +8 s_2 (s_1+s_2) (r_2-r_3)^2+8s_1 (s_1+s_2) (r_1-r_3)^2\\[.1in]
&=&8(s_1+s_2)(s_3-s_1)(s_3-s_2)(s_3-s_1-s_2)\\
&& +8(s_2 (r_2-r_3)+s_1(r_1-r_3))^2\\[.1in]
&>& 0,
\end{array}
\end{equation}
so there is a relative minimum.
\end{proof}

\begin{lemma}
If $\nu_1^{(i)} >0,$ then the four real points $(\nu_1^{(i)},\nu_2^{(i)}),$ $i = 1, \ldots, 4,$ given in Theorem 6, satisfying $\hat{\Sigma}_3=0$ and $\hat{\Sigma}_2^2 - 4 \hat{\Sigma}_4 = 0,$ are singular points of the one-dimensional algebraic set defined by equation~(\ref{algebraicset}),
\begin{equation}
 16 \hat{\Sigma}_4 (\hat{\Sigma}_2^2-4 \hat{\Sigma}_4)^2 - \hat{\Sigma}_3^2 (4 \hat{\Sigma}_2^3-144\hat{\Sigma}_2 \hat{\Sigma}_4 + 27 \hat{\Sigma}_3^2)=0.
\label{nodeequation}
\end{equation}
In particular, at $(\nu_1, \nu_2)=(\nu_1^{(i)},\nu_2^{(i)}),$ for $i=1,2, \mbox{and } 3,$ if $\hat{\Sigma}_2 =-2 \sqrt{\hat{\Sigma}_4} < 0,$  then $(\nu_1^{(i)},\nu_2^{(i)})$ is a node, but if $\hat{\Sigma}_2 = 2 \sqrt{\hat{\Sigma}_4}>0,$ then $(\nu_1^{(i)},\nu_2^{(i)})$ is an isolated point. 
However, at $(\nu_1, \nu_2)=(\nu_1^{(4)},\nu_2^{(4)}),$ $\hat{\Sigma}_2 =-2 \sqrt{\hat{\Sigma}_4} < 0,$  and $(\nu_1^{(4)},\nu_2^{(4)})$ is always a node.

\end{lemma}
\begin{proof}
When $\nu_1 >0,$ equation~(\ref{nodeequation}) defines an algebraic set given by a polynomial equation $Q(\nu_1,\nu_2)=0.$ At $(\nu_1,\nu_2) = (\nu_1^{(4)},\nu_2^{(4)})$, substitution and simplification, using a computer algebra system such as Maple, shows that $Q=0, Q_{\nu_1} = 0, Q_{\nu_2}=0,$ and 
\begin{equation}
 Q_{\nu_1\nu_1} Q_{\nu_2\nu_2} - Q_{\nu_1\nu_2}^2  =  K s_1^2 s_2^2 s_3^2 (s_1+s_2+s_3)^{10} a_1^2 a_2^2 a_3^2 \, \hat{\Sigma}_2^5,
\label{Qeqn}
\end{equation}
where $K>0$ is a positive constant and
\begin{equation}
\begin{array}{rcl}
a_1 & = & (r_2-r_3)^2+ (s_2 + s_3)^2,\\
a_2 & = & (r_1 - r_3)^2 + (s_1 + s_3)^2,\\
a_3 & = & (r_2 - r_1)^2 + (s_1+s_2)^2.
\end{array}
\label{a123equation}
\end{equation}
Clearly $a_1, a_2, $ and $a_3$ are all strictly positive. 
In Lemma~\ref{lemmanu4}, it was shown that $\hat{\Sigma}_2 = - 2 \hat{\Sigma}_4 <0 $ at  $(\nu_1,\nu_2) = (\nu_1^{(4)},\nu_2^{(4)}).$ 
Hence, at $(\nu_1,\nu_2) = (\nu_1^{(4)},\nu_2^{(4)})$
\begin{equation}
 Q_{\nu_1\nu_1} Q_{\nu_2\nu_2}-Q_{\nu_1\nu_2}^2  < 0
\end{equation}
and there is always a node at this singular point.  

In the case of each of the other singular points $(\nu_1^{(i)},\nu_2^{(i)}),$ $i=1,2,3,$ equations~(\ref{Qeqn}) and~(\ref{a123equation}) remain the same, except that the sign in front of $s_i$ changes from positive to negative.  Thus the sign of the expression on the right-hand side of equation~(\ref{Qeqn}) is always the opposite of
the sign of $\hat{\Sigma}_2.$  Therefore, when $\hat{\Sigma}_2 >0,$ there is an isolated point, but when $\hat{\Sigma}_2<0,$ there is a node.
\end{proof}

\begin{lemma}
The common root $z_c$ of  equations~(\ref{cubiczequation}) and~(\ref{newcubic}) that exists at each point of the algebraic set~(\ref{algebraicset}), provided
$\hat{\Sigma}_3 \neq 0,$ has
a positive limit $z = -\hat{\Sigma}_2 > 0$ at the node $(\nu_1^{(4)},\nu_2^{(4)}),$ where $\hat{\Sigma}_3 = 0$ and $\hat{\Sigma}_2 =-2 \sqrt{\hat{\Sigma}_4} < 0.$  Moreover, the same point $(\nu_1^{(4)},\nu_2^{(4)})$ is a regular point of each of the two distinct algebraic sets $\rho_+=0$ and $\rho_-=0$ which cross transversely as subsets of the algebraic set~(\ref{algebraicset}) at the node at $(\nu_1^{(4)},\nu_2^{(4)}).$
\end{lemma}
\begin{proof}
Lemma~\ref{lemmanu4} shows that $\hat{\Sigma}_2 < 0$ at $(\nu_1^{(4)},\nu_2^{(4)}).$  Since $\hat{\Sigma}_3 = 0,$ also, at this point, equations~(\ref{cubiczequation}) and~(\ref{newcubic}) have two common roots, $z=0$ and $z=-\hat{\Sigma}_2.$  Therefore,
the common root $z_c$ of  equations~(\ref{cubiczequation}) and~(\ref{newcubic}) that exists at each point of the algebraic set~(\ref{algebraicset}), provided
$\hat{\Sigma}_3 \neq 0,$ must approach either $0$ or $-\hat{\Sigma}_2.$  The two possible expressions for this common root, given in Theorem~\ref{boundarytheorem}, approach $-\hat{\Sigma}_2$ and $\f{1}{3} \hat{\Sigma}_2,$ respectively.  Since $\hat{\Sigma}_2 <0$ at  $(\nu_1^{(4)},\nu_2^{(4)}),$ the limiting value must be $z = -\hat{\Sigma}_2 > 0.$

Direct calculation shows that, at $(\nu_1^{(4)},\nu_2^{(4)}),$
\begin{equation}
\begin{array}{rcl}
\hat{\Sigma}_{3,\nu_{2}} &=& \f{1}{2 (s_1+s_2+s_3)^4} ((s_1+s_2)(s_1+s_3)(s_2+s_3)(s_1+s_2+s_3)\\
&&+s_1s_2(r_1-r_2)^2+s_1s_3(r_1-r_3)^2+s_2s_3(r_2-r_3)^2 )\\[.1in]
&>& 0,
\end{array}
\end{equation}
and
\begin{equation}
\rho_{\pm,\nu_2} = \hat{\Sigma}_{2,\nu_2} + z_{\nu_2} \pm \f{1}{\sqrt{z}} \hat{\Sigma}_{3,\nu_2}.
\end{equation}
So at least one of $\rho_{\pm,\nu_2}$ is nonzero.  Thus $(\nu_1^{(4)},\nu_2^{(4)})$ is a regular point of at least one of $\rho_+=0$ or $\rho_-=0.$  Since the same point is also a node of the algebraic set~(\ref{algebraicset}), in a neighborhood of which  $z_c>0$ and $\rho_+ \rho_- = 0$ on the algebraic set, the two transverse one-dimensional submanifolds of the algebraic set at the node must correspond to $\rho_+=0$ and $\rho_-=0.$

\end{proof}

\begin{theorem}
The set of possible initial conditions for the Dirichlet eigenvalues $\mu_1$ and $\mu_2$  for smooth two-phase solutions $p=p(x,t)$ of the NLS equation~(\ref{nls}) is parametrized by a single  compact connected region of $\{(\nu_1,\nu_2) \in \mathbb{R}^2\}.$   The maximum value of $\nu_1 =|p(x,t)|^2$ on this region is $(s_1+s_2+s_3)^2,$ and the minimum value of $\nu_1=|p(x,t)|^2$ on this region is
$(s_1+s_2-s_3)^2 > 0,$ if $s_3 > s_1+s_2,$ or zero, if $s_3 \leq s_1+s_2.$
\label{stheorem}
\end{theorem}
\begin{proof}
The previous lemma shows that $\rho_+=0$ and $\rho_- = 0$ cross transversely at the point $(\nu_1,\nu_2) = (\nu_1^{(4)},\nu_2^{(4)}),$ where $z>0.$  These two one-dimensional algebraic sets divide a neighborhood of $(\nu_1,\nu_2) = (\nu_1^{(4)},\nu_2^{(4)})$ into four regions, exactly one of which (to the left of $(\nu_1,\nu_2) = (\nu_1^{(4)},\nu_2^{(4)}),$ since $\nu_1^{(4)}$ must be a relative maximum) contains an open set where $\rho_+ > 0,$ $\rho_- > 0,$ and $z>0.$  Thus there exists a region of initial conditions for $\nu_1$ and $\nu_2$ and, hence, for the Dirichlet eigenvalues $\mu_1$ and $\mu_2,$ on which 
all the reality conditions are satisfied.  Since the smooth bounded solution constructed in this region must reach a maximum value of $\nu_1 = |p(x,t)|^2,$ and
$\nu_1 = \nu_1^{(4)}$ is the only such maximum possible, the closure of the connected region where $\rho_+ >0$ and $\rho_- > 0,$ to the left of   $(\nu_1,\nu_2) = (\nu_1^{(4)},\nu_2^{(4)})$ and bounded by the algebraic set~(\ref{algebraicset}), must be the unique compact connected region consisting of all the permissible initial conditions of $(\nu_1,\nu_2).$ This region parametrizes the permissible initial conditions for $\mu_1$ and $\mu_2$ through the explicit expressions of the first and second symmetric polynomial constraints given in Corollaries 3 and 4, respectively.

Finally, all smooth  two-phase solutions of the NLS equation~(\ref{nls}) are bounded, by Lemma~\ref{boundedlemma}, and so they must achieve a minimum and a maximum value for $\nu_1 = |p(x,t)|^2.$
Therefore all such solutions must oscillate between the maximum value and the minimum value of $\nu_1 = |p(x,t)|^2$ implied by
Theorem 6 and Lemmas 17, 18, and 19. 
\end{proof}
 The simple dependence 
of the extrema of the modulus of the two-phase solution on the branch points is consistent with 
 the analogous result for genus-zero and genus-one solutions~\cite{kamc 00} and numerical simulations  of genus-two solutions~\cite{el 15, osbo 10}.
For higher-phase solutions, analogous results  require the inclusion of higher-time flows in the integrable NLS hierarchy.  The spatial and temporal flows of the scalar NLS equation are not sufficient to span the higher-dimensional torus of three-phase or  higher-phase solutions. In general, the sum of the imaginary parts of the branch points of the 
higher-genus Riemann surface will be an upper bound for the maximum modulus of the solution.

\section{Two-phase solutions of the NLS equation}
In this section, the smooth  two-phase solutions are constructed for  the NLS equation~(\ref{nls}).  The solutions for the 
Dirichlet eigenvalues of such solutions were shown, in the previous section, to be a single compact connected two-dimensional manifold.
As is well-known, the Dubrovin equations for the Dirichlet eigenvalues linearize via the Abel map onto the Jacobi variety of $\mathscr{K}_2,$
\begin{equation}
\begin{array}{rcl}
4i t + d_1 &= &\int\limits^{\mu_1}_{\lambda_1} \f{d \mu_1}{\sqrt{\mathscr{R}(\mu_1)}} + \int\limits^{\mu_2}_{\lambda_3} \f{d \mu_2}{\sqrt{\mathscr{R}(\mu_2)}},\\[.2in]
2ix-4i c_2  t + d_2 & = & \int\limits^{\mu_1}_{\lambda_1} \f{\mu_1 d \mu_1}{\sqrt{\mathscr{R}(\mu_1)}}+\int\limits^{\mu_2}_{\lambda_3} \f{\mu_2 d\mu_2}{\sqrt{\mathscr{R}(\mu_2)}},
\end{array}
\label{abellinearized}
\end{equation}
where
\begin{equation}
\begin{array}{rcl}
d_1 & = & \int\limits^{\mu_{10}}_{\lambda_1} \f{d \mu_1}{\sqrt{\mathscr{R}(\mu_1)}} + \int\limits^{\mu_{20}}_{\lambda_3} \f{d \mu_2}{\sqrt{\mathscr{R}(\mu_2)}},\\[.2in]
d_2 & = & \int\limits^{\mu_{10}}_{\lambda_1} \f{\mu_1 d \mu_1}{\sqrt{\mathscr{R}(\mu_1)}}+\int\limits^{\mu_{20}}_{\lambda_3} \f{\mu_2 d\mu_2}{\sqrt{\mathscr{R}(\mu_2)}},
\end{array}
\label{dequation}
\end{equation}
where the constants of integration $d_1$ and $d_2$ are not arbitrary but are determined by an allowed pair of initial values $\mu_{10}$ and $\mu_{20}$ which satisfy, along with the initial values of the real parameters $\nu_1$ and $\nu_2,$ the algebraic constraints imposed by the invariant algebraic curve.  Moreover, this set of permissible values for $d_1$ and $d_2$ is partitioned into equivalence classes  lying on distinct non-intersecting two-real dimensional planes in $\mathbb{C}^2$ formed 
by translation along the linear flows of equation~(\ref{abellinearized}) on the Jacobian variety  of $\mathscr{K}_2.$

Solving equations~(\ref{abellinearized}) for $\mu_1$ and $\mu_2$ is the classical Jacobi inversion problem.  The details of solving the Jacobi inversion problem, in terms of the Kleinian elliptic functions $\sigma$ and $\zeta$ defined on the genus-two Riemann surface $\mathscr{K}_2,$ are placed in the Appendix.  If $u^\prime = (u_1^\prime, u_2^\prime)^T,$ with
\begin{equation}
\begin{array}{rclcl}
u_1^\prime & = & \int\limits^{\mu_1}_{\lambda_1} \f{d \mu_1}{\sqrt{\mathscr{R}(\mu_1)}} + \int\limits^{\mu_2}_{\lambda_3} \f{d \mu_2}{\sqrt{\mathscr{R}(\mu_2)}}
& = & 4 i t +d_1,\\[.2in]
u_2^\prime & = &  \int\limits^{\mu_1}_{\lambda_1} \f{\mu_1 d \mu_1}{\sqrt{\mathscr{R}(\mu_1)}}+\int\limits^{\mu_2}_{\lambda_3} \f{\mu_2 d\mu_2}{\sqrt{\mathscr{R}(\mu_2)}} & = & 2 i x +2 i \Lambda_1 t + d_2,
\end{array}
\end{equation}
so that 
\begin{equation}
u^\prime =  i V x + i W t + d,
\end{equation}
where
\begin{equation}
\begin{array}{rcl}
 V & = & (0,2)^T, \\
  W & = & (4, 2 \Lambda_1)^T,\\
d & = & (d_1, d_2)^T,
\end{array}
\end{equation}
and $d_1$ and $d_2$ are defined in equation~(\ref{dequation}),
then equation~(\ref{implicitsolution}) produces simple formulas for the symmetric polynomials of $\mu_1$ and $\mu_2.$ In particular,
\begin{equation}
\begin{array}{rcl}
\mu_1 + \mu_2 & =& \f{1}{2}\Lambda_1 -\delta_2 + \f{\partial}{\partial u_2} \ln \left( \f{\sigma (u^\prime + \int^{\infty^+}_{\lambda_6} du)}{\sigma (u^\prime + \int^{\infty^-}_{\lambda_6} du)}\right), \\[.2in]
\mu_1 \mu_2 & = &  -\f{1}{8} \Lambda_1^2+\f{1}{2} \Lambda_2+\delta_1 - \f{\partial}{\partial u_1} \ln  \left( \f{\sigma (u^\prime + \int^{\infty^+}_{\lambda_6} du)}{\sigma (u^\prime + \int^{\infty^-}_{\lambda_6} du)}\right),
\end{array}
\label{musolution}
\end{equation}
where $\delta_1$ and $\delta_2$ are constants of integration expressible in terms of functions of the branch points.
The solution $p$ to the NLS equation~(\ref{nls}) is obtained from the trace formulas in Equation~(\ref{trace}) and Equation~(\ref{musolution}), viz.,
\begin{equation}
\begin{array}{rcl}
\f{\partial}{\partial x} \ln p & =& 2 i (\mu_1 + \mu_2) - i \Lambda_1\\[.2in]
& = & 2 i  \f{\partial}{\partial u_2} \ln \left( \f{\sigma (u^\prime + \int^{\infty^+}_{\lambda_6} du)}{\sigma (u^\prime + \int^{\infty^-}_{\lambda_6} du)}\right) - 2 i \delta_2 \\[.2in]
& = & \f{\partial}{\partial x}   \ln \left( \f{\sigma (u^\prime + \int^{\infty^+}_{\lambda_6} du)}{\sigma (u^\prime + \int^{\infty^-}_{\lambda_6} du)}\right)-2 i \delta_2,\\[.2in]
\f{\partial}{\partial t} \ln p & = & -4 i \mu_1 \mu_2 + 2 i \Lambda_2 - \f{1}{2}i \Lambda_1^2+ \Lambda_1 \f{\partial}{\partial x} \ln p \\[.2in]
&= & 4 i \f{\partial}{\partial u_1}  \ln \left( \f{\sigma (u^\prime + \int^{\infty^+}_{\lambda_6} du)}{\sigma (u^\prime + \int^{\infty^-}_{\lambda_6} du)}\right)  + 2i\Lambda_1
 \f{\partial}{\partial u_2} \ln \left( \f{\sigma (u^\prime + \int^{\infty^+}_{\lambda_6} du)}{\sigma (u^\prime + \int^{\infty^-}_{\lambda_6} du)}\right) - 4 i \delta_1 - 2 i \delta_2 \Lambda_1\\[.2in]
& = &\f{\partial}{\partial t}  \ln \left( \f{\sigma (u^\prime + \int^{\infty^+}_{\lambda_6} du)}{\sigma (u^\prime + \int^{\infty^-}_{\lambda_6} du)}\right) - 4 i \delta_1  - 2 i \delta_2 \Lambda_1.
\end{array}
\end{equation}
Hence, the two-phase solution to the NLS equation~(\ref{nls}) is
\begin{equation}
\begin{array}{rcl}
p(x,t) &=& \sqrt{\nu_1 (0,0)} e^{i \phi(x,t)} \f{\sigma \left(\beta^- \right)}{\sigma \left(\beta^+ \right)} \f{\sigma \left(i V x + i W t +\beta^+ \right)}{\sigma \left(i V x + i W t + \beta^-\right)} 
\end{array}
\label{twosigma}
\end{equation}
where
\begin{equation}
\begin{array}{rcl}
\beta^+ & = & \int^{\mu_{10}}_{\lambda_1} du + \int^{\mu_{20}}_{\lambda_3} du + \int^{\infty^+}_{\lambda_6} du, \\[.1in]
\beta^- & = & \int^{\mu_{10}}_{\lambda_1} du +\int^{\mu_{20}}_{\lambda_3} du + \int^{\infty^-}_{\lambda_6} du,\\[.1in]
\phi(x,t) & = & -2  \delta_2 x - (4  \delta_1 + 2 \delta_2 \Lambda_1)t + \phi_0,
\end{array}
\label{betadefinition}
\end{equation}
$\delta_1, \delta_2 \in \mathbb{C}$ are given explicitly by equation~(\ref{deltaeqn}), and $\phi_0 \in \mathbb{R}$ is an arbitrary real constant. Note that $\delta_1$ and $\delta_2$  in $\phi$ are, in general, complex because there is also a complex phase in the ratio of the $\sigma$ functions.  These complex phases cancel out, so that the solution remains bounded for all $x$ and $t.$
The values of $\nu_1 (0,0) >0$ and $(\mu_{10},s_{10}), (\mu_{20},s_{20}) \in \mathscr{K}_2$ can be explicitly given in terms of the branch points of the curve $\mathscr{K}_2$  by the formulas derived in the previous section.

It is convenient to remove the exponential factors from the $\sigma$ functions in equation~(\ref{twosigma}) and write the two-phase solution in terms of Riemann theta functions.  The definitions of the theta functions and the constants in the following expressions are given in the Appendix.
\begin{equation}
\begin{array}{rcl}
p(x,t) &=& \sqrt{\nu_1 (0,0)} e^{i \xi(x,t) } \f{\theta \left(\f{1}{2}\omega^{-1}\beta^- \right)}{\theta \left(\f{1}{2} \omega^{-1} \beta^+ \right)} \f{\theta \left(\f{1}{2} \omega^{-1} (i V x + i W t + \beta^+) \right)}{\theta \left(\f{1}{2}\omega^{-1} (i V x + i W t + \beta^- )\right)}, 
\end{array}
\label{firstp}
\end{equation}
where
\begin{equation}
\begin{array}{rcl}
\xi(x,t) & = & \phi(x,t) + (2\alpha_{11} 4 t+(\alpha_{12}+\alpha_{21})(2 x + 2 \Lambda_1 t)) \int^{\infty^+}_{\infty^-} du_1 \\
    &  &+ (2 \alpha_{22} (2x+2\Lambda_1 t)+ (\alpha_{12}+\alpha_{21}) 4t) \int^{\infty^+}_{\infty^-} du_2. 
\end{array}
\end{equation}

\begin{lemma}  The exponential factor $e^{i\xi(x,t)}$ in equation~(\ref{firstp}) is a phase factor of modulus 1, viz.,
\begin{equation}    
 \xi(x,t)= \kappa_1 x + \kappa_2 t + \phi_0,
\label{xidefinition}
\end{equation}
where the wavenumbers $\kappa_1,\kappa_2 \in \mathbb{R}$ are given by
\begin{equation}
\begin{array}{rcl}
\kappa_1 & = & 2(\lambda_1+\lambda_1^* -\f{1}{2}\Lambda_1) - (\omega^{-1})_{12}(\f{\theta_1^+}{\theta^+}-\f{\theta_1^-}{\theta^-})- (\omega^{-1})_{22}(\f{\theta_2^+}{\theta^+}-\f{\theta_2^-}{\theta^-}),      \\[.1in]
\kappa_2 & = &  4 (- |\lambda_1|^2 - \f{1}{8} \Lambda_1^2 + \f{1}{2} \Lambda_2)  -2(\omega^{-1})_{11} (\f{\theta_1^+}{\theta^+}-\f{\theta_1^-}{\theta^-})-2  (\omega^{-1})_{21} (\f{\theta_2^+}{\theta^+} - \f{\theta_2^-}{\theta^-})\\
&&+\Lambda_1 \left( 2 (\lambda_1 + \lambda_1^* -\f{1}{2} \Lambda_1) 
 - (\omega^{-1})_{12}(\f{\theta_1^+}{\theta^+}-\f{\theta_1^-}{\theta^-})- (\omega^{-1})_{22}(\f{\theta_2^+}{\theta^+}-\f{\theta_2^-}{\theta^-}) \right),
\end{array}
\end{equation}
and $\phi_0 \in \mathbb{R}$ is an arbitrary phase.
\end{lemma}
\begin{proof}
Since $\omega^{-1}$ is purely imaginary (see the Appendix for details), it is sufficient to show that the differences of logarithmic derivatives of $\theta$ of the form, for $j=1,2,$
\begin{equation}
\f{\theta_j^+}{\theta^+}-\f{\theta_j^-}{\theta^-}
\end{equation}
are also purely imaginary.  
Consider the complex conjugate
\begin{equation}
\begin{array}{rcl}
\left(\f{1}{2} \omega^{-1}( \int^{\infty^+}_{\lambda_6} du + \int^{\lambda_2}_{\lambda_3} du)\right)^* & = & - \f{1}{2} \omega^{-1} ( \int^{\infty^+}_{\lambda_5} du + \int^{\lambda_1}_{\lambda_4} du)\\[.1in]
& = & - \f{1}{2} \omega^{-1} ( \int^{\infty^+}_{\lambda_6} du + \int^{\lambda_2}_{\lambda_3} du) \\
&&- v^{\lambda_6,\lambda_5}-
v^{\lambda_1,\lambda_2}- v^{\lambda_3,\lambda_4}\\[.1in]
&=&  - \f{1}{2} \omega^{-1} ( \int^{\infty^+}_{\lambda_6} du + \int^{\lambda_2}_{\lambda_3} du)  + m + \tau m^\prime,
\end{array}
\end{equation}
where the sum of the half-period characteristics is zero, i.e., an integer lattice translation for some $m, m^\prime \in \mathbb{Z}^2.$
Similarly,
\begin{equation}
\left(\f{1}{2} \omega^{-1} \left( \int^{\infty^+}_{\lambda_6} du + \int^{\lambda_2}_{\lambda_3} du\right)\right)^* =  - \f{1}{2} \omega^{-1} \left( \int^{\infty^+}_{\lambda_6} du + \int^{\lambda_2}_{\lambda_3} du\right)  + m + \tau m^\prime,
\end{equation}
with the same $m, m^\prime \in \mathbb{Z}^2$ as previously, since we can integrate along the same paths in the calculation of the half-integer characteristics.

Now using the fact that $\theta$ is an even function and the partial derivatives $\theta_1$ and $\theta_2$ are odd functions, 
and the transformation property of the logarithmic derivatives in equation~(\ref{logtheta}), we obtain, for $j=1,2,$
\begin{equation}
\begin{array}{rcl}
\left( \f{\theta^+_j}{\theta^+}\right)^* & = & -2 \pi i m_j^\prime - \f{\theta_j^{+}}{\theta^+},\\[.1in]
\left( \f{\theta^-_j}{\theta^-}\right)^* & = & -2 \pi i m_j^\prime - \f{\theta_j^{-}}{\theta^-}.
\end{array}
\end{equation}
Hence
\begin{equation}
\left( \f{\theta^+_j}{\theta^+} - \f{\theta^-_j}{\theta^-} \right)^* = - \left( \f{\theta^+_j}{\theta^+} - \f{\theta^-_j}{\theta^-} \right).
\end{equation}
\end{proof}

Finally, all the previous results can be summarized in the following theorem, in which the  solution of the NLS equation~(\ref{nls}) is constructed, explicitly parametrized by the branch points of the Riemann surface $\mathscr{K}_2,$ with all reality conditions  satisfied.  Moreover simple formulas for the minimum and the maximum of the solution are determined, based solely on the imaginary parts of the branch points of $\mathscr{K}_2.$
\begin{theorem}
The smooth two-phase solutions of the NLS equation~(\ref{nls}) with non-singular invariant curve $\mathscr{K}_2,$ 
form a  two-real-dimensional torus submanifold, modulo a circle of complex phase factors, of the Jacobi variety of $\mathscr{K}_2,$ with branch points
$\lambda_1 = \lambda_2^* = r_1 + i s_1,$ $\lambda_3=\lambda_4^* = r_2 + i s_2,$ $\lambda_5 = \lambda_6^* = r_3 + i s_3,$ where $r_i, s_i \in \mathbb{R},$ for $i = 1,2,3,$ and $s_1, s_2, s_3 > 0.$ 
The two-phase solution has maximum modulus $s_1+s_2+s_3>0.$ The  minimum modulus of the solution is $s_3-s_1-s_2>0,$ if $s_3 > s_1 + s_2,$ or zero, if $s_3 \leq s_1+s_2.$  In general, the solution has an explicit representation in terms of the $\theta$ function associated with $\mathscr{K}_2,$ as constructed in the Appendix, and the branch points of $\mathscr{K}_2,$ viz.,
 \begin{equation}
p(x,t) =(s_1+s_2+s_3) e^{i \xi(x-x_0,t-t_0) } \f{\theta \left(\f{1}{2}\omega^{-1}\beta^- \right)}{\theta \left(\f{1}{2} \omega^{-1} \beta^+ \right)} \f{\theta \left(\f{1}{2} \omega^{-1} U^+(x,t) \right)}{\theta \left(\f{1}{2}\omega^{-1} U^-(x,t)\right)}, 
\end{equation}
where
\begin{equation}
\begin{array}{rcl}
U^+(x,t) &= &  i V (x-x_0) + i W (t-t_0) + \beta^+ ,\\[.1in]
U^-(x,t) & = & i V (x-x_0) + i W (t-t_0) + \beta^-,
\end{array}
\end{equation}
with $V  =  (0,2)^T, $ $ W  =  (4, 4(r_1+r_2+r_3))^T,$
\begin{equation}
\begin{array}{rcl}
\beta^+ & = & \int^{(\mu_{1},-\sqrt{\mathscr{R}(\mu_1)})}_{\lambda_1} du + \int^{(\mu_{2},\sqrt{\mathscr{R}(\mu_2)})}_{\lambda_3} du + \int^{\infty^+}_{\lambda_6} du, \\[.1in]
\beta^- & = & \int^{(\mu_{1},-\sqrt{\mathscr{R}(\mu_1)})}_{\lambda_1} du +\int^{(\mu_{2},\sqrt{\mathscr{R}(\mu_2)})}_{\lambda_3} du + \int^{\infty^-}_{\lambda_6} du,
\end{array}
\end{equation}
where the initial values of the Dirichlet eigenvalues $(\mu_{1},-\sqrt{\mathscr{R}(\mu_1)})$ and \newline $(\mu_{2},\sqrt{\mathscr{R}(\mu_2)})$ are given by equations~(\ref{muexplicitatnu4}) and~(\ref{Rexplicitatnu4}), the phase $\xi(x,t)$ is given by  equation~(\ref{xidefinition}), and $(x_0, t_0) \in \mathbb{R}^2$ is
an arbitrary  location for the maximum modulus of the solution.
\label{finaltheorem}
\end{theorem}

\section{Example Solutions}
Two solutions are constructed using the formulas of Theorem~\ref{finaltheorem}. Both solutions are constructed using Maple and a simple integration algorithm, based
on Simpson's rule, to compute the necessary contour integrals on the Riemann surface $\mathscr{K}_2.$  For integrals involving points over infinity or branch points, Maple's algcurve package is used to compute a Puiseux expansion near the point of interest, which is  integrated directly using Maple's int command~\cite{deco 11}.

In Figure~\ref{figsoln1}, the branch points are chosen to be
\begin{equation}
-2+i,-2-i,-1+i,-1-i,1+i,1-i,
\end{equation}
so that the imaginary parts of the branch points are $s_1 =1, s_2=1,$ and $s_3=1.$  Thus the maximum of the solution's modulus is $s_1+s_2+s_3=3.$  Also $s_3 < s_1 + s_2,$ so the minimum is 0.

In Figure~\ref{figsoln2}, the branch points are chosen to be
\begin{equation}
-2+3 i, -2-3i, -1+i,-1-i, 1+i, 1-i,
\end{equation}
so that the imaginary parts of the branch points are $s_1 = 3, s_2=1,$ and $s_3=1.$  Thus the maximum of the solution's modulus is $s_1+s_2+s_3 =5.$ Also 
$s_3> s_1+s_2,$ so the minimum is $s_3-s_1-s_2 = 1.$

\begin{figure}
\centering
\includegraphics[height=8cm]{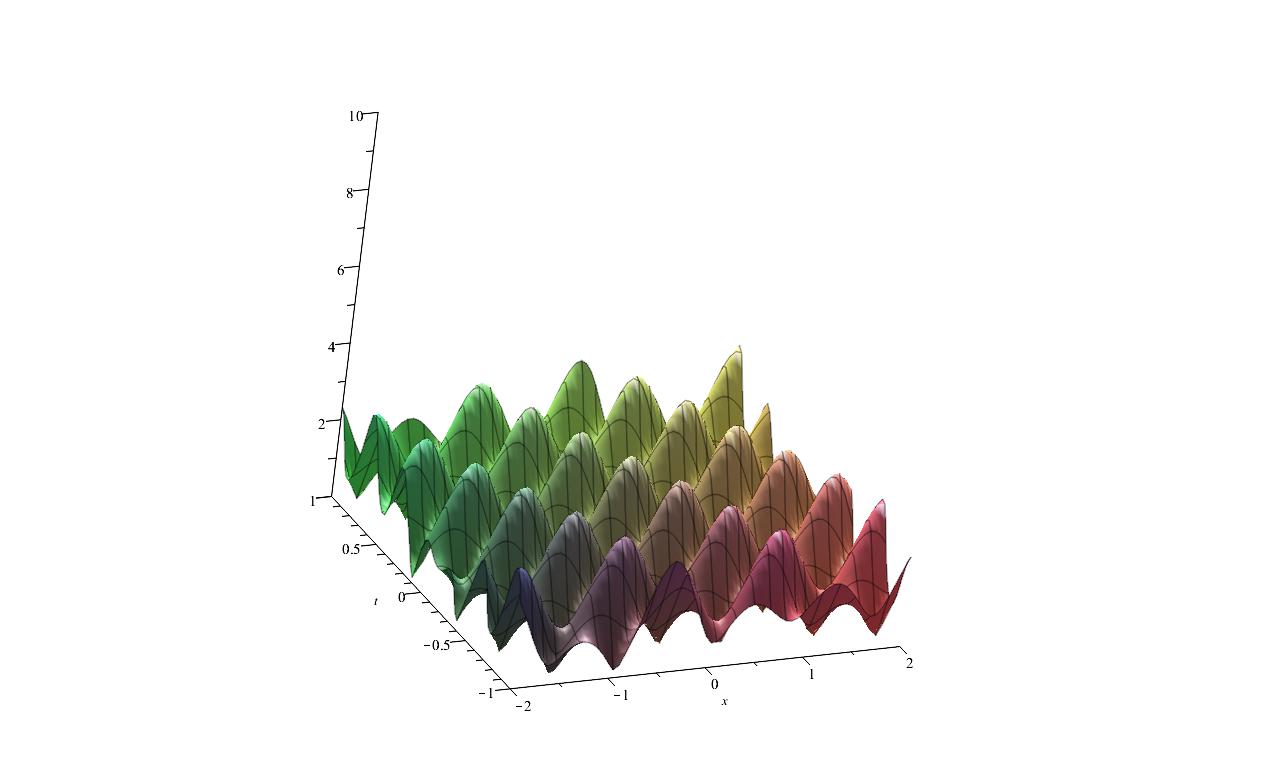}
\caption{Modulus of a solution of the NLS equation~(\ref{nls}) with maximum 3 and minimum 0.}
\label{figsoln1}
\end{figure}

\begin{figure}
\centering
\includegraphics[height=8cm]{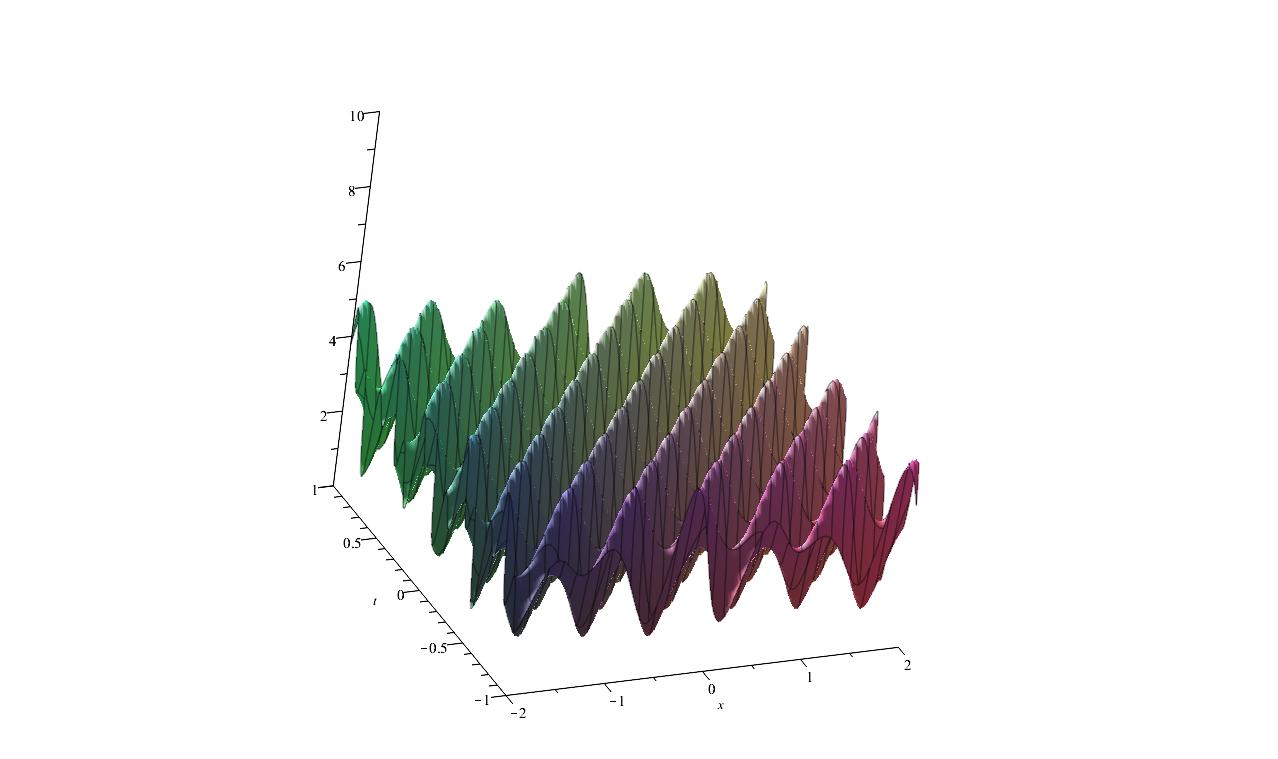}
\caption{Modulus of a solution of the NLS equation~(\ref{nls}) with maximum 5 and minimum 1.}
\label{figsoln2}
\end{figure}

\section{Conclusion}

An explicit parametrization of all smooth two-phase solutions of the scalar cubic  NLS equation~(\ref{nls})  is given in terms of the branch points of the invariant spectral curve and the loci of the Dirichlet eigenvalues.  In particular, simple formulas are found, in terms of the imaginary parts of the branch points, for the maximum modulus and the minimum modulus of each smooth two-phase solution.  Independently, a general theta function formula for the solution is given with parametrization in terms of the branch points of  the genus-two Riemann surface and the values of the Dirichlet eigenvalues which satisfy the reality condtions. 

The dependence of the maximum modulus of the two-phase solution on the imaginary parts of the branch points is consistent with known results for zero-phase
and one-phase solutions~\cite{kamc 00} and with numerical simulations of two-phase solutions~\cite{el 15, osbo 10}. For  higher-phase solutions of the scalar cubic NLS equation~(\ref{nls}), the two-real dimensional linear flow on the Jacobi variety will not span the higher-dimensional real torus of solutions, as it does in the two-dimensional case, so the sum of the imaginary parts of the branch points will be an upper bound on the modulus of the quasi-periodic solution, not necessarily the maximum value.

The parametrization obtained in this manuscript may be useful in characterizing modulations of two-phase solutions, for example, in
the vicinity of a gradient catastrophe in which spikes are formed which are limits of two-phase solutions~\cite{bert 13, el 15}.
More work is needed to extend the current results to the scalar defocusing cubic NLS equation and the Manakov system of coupled NLS equations~\cite{wood 07, kamc 14, elgi 07, wrig 13}.

It is well-known (see~\cite{dubr 81, belo 94} and the references therein) how to obtain the solution of the NLS equation $p(x,t)$ in terms of theta functions defined on the Jacobi variety  of the Riemann surface $\mathscr{K}_2,$ by examining the asymptotic expansion of the  Baker-Akhiezer solution of the Lax pair.  In that approach,  constants of integration are chosen in a manner sufficient to satisfy the reality conditions.  In this paper, necessary and sufficient conditions for the existence of smooth two-phase solutions are derived for all the integration constants, using the explicit representation of the $\mu$-trajectories of the Dirichlet eigenvalues, analogous to the results for elliptic solutions obtained by Kamachatnov~\cite{kamc 00, kamc 90}.  
Consequently, a theta function representation of the solution $p(x,t)$ is also constructed, using the classical integration of the Abelian integrals by the Kleinian ultra-elliptic functions, with all constants of integration  expressed in terms of the branch points of the Riemann surface. 

\renewcommand{\theequation}{A-\arabic{equation}}
  \setcounter{equation}{0}  
  \section{Appendix - The Jacobi Inversion Problem} 
The inversion of the Abelian integrals in equation~(\ref{linearized}) for $\mu_1 = \mu_1 (x,t)$ and $\mu_2 = \mu_2 (x,t)$ is called the Jacobi inversion problem.
Classically~\cite{bake 07} the inversion problem for the symmetric polynomials $\mu_1+\mu_2$ and $\mu_1\mu_2$ was solved in terms of the Kleinian sigma and zeta functions defined on the odd-degree genus-two curves which generalize the more familiar elliptic functions, i.e., the  Weierstrass sigma and zeta functions.  A more modern treatment~\cite{eilb 03} shows how the Kleinian elliptic functions can be used to solve the inversion problem for even-degree hyperelliptic curves such as $\mathscr{K}_2.$

\subsection{Differential Identities on the Riemann Surface}
In order to solve the Jacobi inversion problem~(\ref{linearized}), several definitions and fundamental identities involving differentials on the Riemann surface are needed.  The goal of the present manuscript is to keep technicalities to a minimum and follow as closely as possible the classical
approach of Baker~\cite{bake 07}, while using the more modern language found in~\cite{belo 94} and~\cite{eilb 03}.
First a canonical basis of homology cycles, $\{\mathfrak{a}_1, \mathfrak{a}_2; \mathfrak{b}_1, \mathfrak{b}_2\},$ is constructed on $\mathscr{K}_2$ satisfying the intersection properties $\mathfrak{a}_1 \circ \mathfrak{a}_2 = \mathfrak{b}_1 \circ \mathfrak{b}_2 = 0$ and $\mathfrak{a}_1 \circ \mathfrak{b}_1 = \mathfrak{a}_2 \circ \mathfrak{b}_2 = 1.$  The action of the natural hyperelliptic involution $\iota$ on $\mathscr{K}_2$ is extended in a natural way to the homology cycles,
and the cycles are chosen such that
\begin{equation}
\begin{array}{rcl}
\iota (\mathfrak{a}_1) & = & - \mathfrak{a}_1,\\
\iota(\mathfrak{a}_2) & = & -\mathfrak{a}_2,\\
\iota(\mathfrak{b}_1) & = & -\mathfrak{b}_1,\\
\iota(\mathfrak{b}_2)& = & -\mathfrak{b}_2.
\end{array}
\end{equation}
In particular, the arrangement of cycles is chosen as shown in Figure~\ref{cycles} for the case of three complex-conjugate pairs of branch points which satisfy
the reality conditions for the NLS equation~(\ref{nls}).  Similarly, the natural action of the anti-holomorphic involution $*$ on the cycles is
\begin{equation}
\begin{array}{rcl}
*(\mathfrak{a}_1) & = & -\mathfrak{a}_1, \\
*(\mathfrak{a}_2)& =& -\mathfrak{a}_2,\\
*(\mathfrak{b}_1) & = &  \mathfrak{b}_1 + \mathfrak{a}_2,\\
*(\mathfrak{b}_2) & = & \mathfrak{b}_2 +\mathfrak{a}_1.
\end{array}
\label{conjugatecycles}
\end{equation}
\begin{figure}
\centering
\includegraphics[height=8cm]{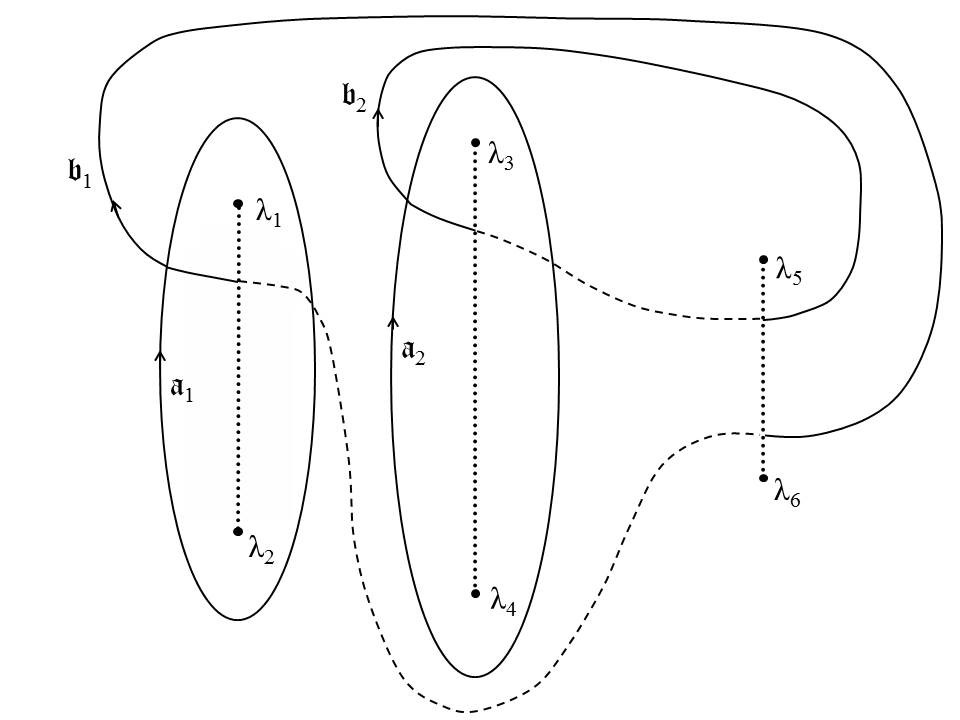}
\caption{Basis of homology cycles, genus two, for focusing NLS equation.}
\label{cycles}
\end{figure}

Now consider the two holomorphic integrals that appear in the Jacobi inversion problem~(\ref{linearized}),
\begin{equation}
du_1  =  \frac{d\lambda}{w}, \hspace{.5in}
du_2  =  \frac{\lambda d\lambda}{w}.
\end{equation}
The natural action of the anti-holomorphic involution on these two differentials is $*du_i  =  du_i^*,$ for $i=1,2,$ where $du_i^*$ denotes the complex conjugate of the differential. 
The periods of the above differentials around the basis cycles are defined, for $i,j = 1,2,$ as
\begin{equation}
\oint_{\mathfrak{a}_j} du_i  =  2 \omega_{ij},\hspace{.5in} \oint_{\mathfrak{b}_j} du_i  =  2 \omega^\prime_{ij}.
\end{equation}
Notice that for the given choice of the basis cycles,
\begin{equation}
\omega_{ij}^* = \oint_{\mathfrak{a}_j} du_i^* = \oint_{\mathfrak{a}_j} *du_i = \oint_{*(\mathfrak{a}_j) }du_i = \oint_{-\mathfrak{a}_j} du_i = -\omega_{ij},
\end{equation}
so that $\Re (\omega_{ij}) =0.$  

A normalized basis of holomorphic differentials, canonically dual to the basis of homology cycles, can now be constructed.  Let, for $i=1,2,$
\begin{equation}
\begin{array}{rcl}
d v_i & = & \sum\limits_{j=1}^2 \f{1}{2}(\omega^{-1})_{ij} du_j,
\end{array}
\end{equation}
then the periods are normalized in the sense that, for $i,j=1,2,$
\begin{equation}
\oint_{\mathfrak{a}_j} dv_i  =  \delta_{ij},\hspace{.5in}
\oint_{\mathfrak{b}_j} dv_i  =  \tau_{ij},
\end{equation}
where $\delta_{ij}$ is the Kronecker delta symbol.  A standard result, using Riemann's bilinear relations, is that the determinant of the matrix $\omega = (\omega_{ij})$ is nonzero, and the matrix
\begin{equation}
\tau = \omega^{-1} \omega^\prime
\label{tau}
\end{equation}
is symmetric with a positive definite imaginary part.

With the particular choice of the canonical cycles shown in Figure~\ref{cycles}, the symmetry of $\mathscr{K}_2$ implies further information about the real part $\Re (\tau)$ of $\tau,$ viz., 
\begin{equation}
\begin{array}{rcl}
\tau_{ij}^* & = & \oint\limits_{\mathfrak{b}_j} dv_i^* \\
&=& \oint\limits_{\mathfrak{b}_j} \sum\limits_{k=1}^2 \f{1}{2}(\omega^{-1})^*_{ik} du_k^*\\
& = &-\sum\limits_{k=1}^2\f{1}{2} (\omega^{-1})_{ik} \oint\limits_{\mathfrak{b}_j} *du_k\\
& = &-\sum\limits_{k=1}^2\f{1}{2} (\omega^{-1})_{ik} \oint\limits_{*(\mathfrak{b}_j)} du_k\\
& = &-\sum\limits_{k=1}^2\f{1}{2} (\omega^{-1})_{ik} \oint\limits_{\mathfrak{b}_j+\hat{\mathfrak{a}}_{j}} du_k\\
& = &-\sum\limits_{k=1}^2\f{1}{2} (\omega^{-1})_{ik}(2 \omega^\prime_{kj} + 2 \hat{\omega}_{kj})\\
& =& -\tau_{ij}- \hat{\delta}_{ij},
\end{array}
\end{equation}
where $\hat{\mathfrak{a}}_1 = \mathfrak{a}_2$ and $\hat{\mathfrak{a}}_2=\mathfrak{a}_1,$ so that the first and second columns of  $\hat{\omega}$ are the same as the second and first columns, respectively, of $\omega,$ and 
\begin{equation}
\hat{\delta}_{ij} = \left(\begin{array}{cc} 0 & 1 \\ 1 & 0 \end{array} \right).
\end{equation}
Therefore
\begin{equation}
\Re (\tau) = -\f{1}{2}\left(\begin{array}{cc} 0 & 1 \\ 1 & 0 \end{array} \right).
\label{realtau}
\end{equation}

We now introduce two differentials of the second kind associated with $du_1$ and $du_2,$ see~\cite{bake 07}, viz., 
\begin{equation}
dr_1  = \f{4 \lambda^4 -3\Lambda_1 \lambda^3+2 \Lambda_2 \lambda^2-\Lambda_3 \lambda}{4 w}d\lambda,\hspace{.4in}
dr_2  =  \f{2 \lambda^3 - \Lambda_1 \lambda^2}{4 w} d\lambda.
\end{equation}
These differentials are holomorphic except for poles of the second kind at the points at infinity, and they satisfy the following identity of two-differentials,
\begin{equation}
\frac{\partial}{\partial z} \frac{w+s}{\lambda-z} \frac{1}{2 w} dz\,d\lambda + du_1(\lambda)dr_1(z) + du_2(\lambda) dr_2(z) = \frac{F(\lambda,z)+ 2 w s}{4(\lambda-z)^2} \frac{d\lambda}{w}\frac{dz}{s},
\label{twoformidentity}
\end{equation}
where $w^2 = \mathscr{R}(\lambda), s^2 = \mathscr{R}(z),$ and the symmetric function  $F(\lambda,z)=F(z,\lambda),$ where
\begin{equation}
F(\lambda,z) =  2 \Lambda_6 - \Lambda_5(\lambda+z) + \lambda z (2 \Lambda_4 - \Lambda_3 (\lambda+z))  +\lambda^2 z^2
(2 \Lambda_2 - \Lambda_1 (\lambda+z)) + 2 \lambda^3 z^3.
\end{equation}
The periods of the differentials $dr_1$ and $dr_2$ are defined, for $i, j = 1, 2,$  as
\begin{equation}
\oint_{\mathfrak{a}_j} dr_i = -2\eta_{ij}, \hspace{.5in} \oint_{\mathfrak{b}_j} dr_i= -2\eta^\prime_{ij}.
\end{equation}
The periods satisfy the following relations,
\begin{equation}
\begin{array}{ccc}
\omega \omega^{\prime\, T}-\omega^\prime \omega^T &=& 0,\\
\eta \eta^{\prime \,T} - \eta^\prime \eta^T & = & 0,\\
\omega \eta^{\prime \, T} - \omega^\prime \eta^T & = & - \frac{\pi i}{2} \mbox{\bf I},
\end{array}
\end{equation}  
where $\omega^T$ denotes the transpose of the matrix $\omega,$ and $\mbox{\bf I}$ is the $2 \times 2$ identity matrix.
Equation~(\ref{twoformidentity}) shows that the symmetric two-differential on the right-hand side of the equation gives Klein's symmetric
integral of the third kind,
\begin{equation}
\mathscr{K}(\lambda,a;z,b) = \int\limits_a^\lambda \int\limits_b^z \frac{F(\lambda,z)+2 w s}{4(\lambda-z)^2}\frac{d\lambda}{w} \frac{dz}{s},
\label{Kdefinition}
\end{equation}
with logarithmic infinities of coefficients $1$ and $-1$ respectively at $\lambda=z$ and $\lambda=b.$  Notice that
the symmetry of the integrand implies that $\mathscr{K} (\lambda,a;z,b) = \mathscr{K} (z,b;\lambda,a).$

\begin{definition}
The normalized differential of the third kind  having a simple pole of residue $+1$ at the location $\lambda=z,$ simple pole with residue $-1$ at the location $\lambda=b,$ and zero periods around the cycles $\mathfrak{a}_1$ and $\mathfrak{a}_2$ is denoted by $d\Pi_{z,b}.$  The normalized differential of the second kind having a single pole of order two at $\lambda=z$ with coefficient $1$
and zero periods around the cycles $\mathfrak{a}_1$ and $\mathfrak{a}_2$ is denoted by $d\Gamma_z.$
\end{definition}
Notice that $d\Gamma_z$ can be obtained from $d\Pi_{z,b}$ by differentiation of the latter with respect to the parameter $z.$  Thus an application of Riemann's bilinear relations using $dv_{i},$ $i=1,2,$ and $d\Pi_{z,b}$ shows that 
\begin{equation}
\oint_{\mathfrak{b}_1} d\Gamma_z  =  2 \pi  i \f{dv_1}{dz}(z), \hspace{.3in} \oint_{\mathfrak{b}_2} d\Gamma_z  =  2 \pi i \f{dv_2}{dz}(z).
\end{equation}
Another application of Riemann's bilinear relations to $d\Pi_{\lambda,a}$ and $d\Pi_{z,b},$ demonstrates the following symmetry,
\begin{equation}
\int_a^\lambda d\Pi_{z,b} = \int_b^z d\Pi_{\lambda,a}.
\label{Pisymmetry}
\end{equation}
Since $\mathscr{K}(\lambda,a;z,b)$ is an integral of the third kind, it can be written in terms of the normalized integral of the third kind and two linearly independent integrals of the first kind, leading to the following definition, in which the symmetry of $\mathscr{K}$ and equation~(\ref{Pisymmetry}) are used.
\begin{definition}
Let $\alpha$ be the $2 \times 2$ symmetric matrix defined by the identity
\begin{equation}
\mathscr{K} (\lambda,a;z,b) = \int_a^\lambda d\Pi_{z,b} - 2 \sum\limits_{i=1}^2\sum\limits_{j=1}^2 \alpha_{ij} \int_a^\lambda du_i \int_b^z du_j.
\label{alphadefinition}
\end{equation}
\end{definition}
Differentiation of the above identity with respect to local parameters $\lambda$ and $z$ gives the following identity between two-differentials,
\begin{equation}
\frac{F(\lambda,z)+ 2 w s}{4(\lambda-z)^2} \frac{d\lambda}{w}\frac{dz}{s} = d\Gamma_z (\lambda) dz  - 2 \sum\limits_{i=1}^2\sum\limits_{j=1}^2 \alpha_{ij} du_i(\lambda)  du_j (z).
\label{alphaidentity}
\end{equation}
If the roles of $\lambda$ and $z$ are interchanged in equation~(\ref{twoformidentity}), using the symmetry of the two-form on the right-hand side, then equation~(\ref{alphaidentity}), can be re-written as 
\begin{equation}
 du_1(z)dr_1(\lambda) + du_2(z) dr_2(\lambda) + 2 \sum\limits_{i=1}^2\sum\limits_{j=1}^2 \alpha_{ij}   du_i (\lambda) du_j(z)= d\Gamma_z (\lambda) dz -\frac{\partial}{\partial \lambda} \frac{s+w}{z - \lambda} \frac{d\lambda dz}{2s}.
\label{newidentity}
\end{equation}
Integration in $\lambda$ of Equation~(\ref{newidentity}) about the basis cycles $\mathfrak{a}_k$ and $\mathfrak{b}_k,$ for $k=1,2,$ gives two identities of the following differentials of $z,$
\begin{equation}
\begin{array}{rcl}
- 2\eta_{1k} du_1 -2 \eta_{2k} du_2+4\sum\limits_{i=1}^2\sum\limits_{j=1}^2 \alpha_{ij} \, \omega_{ik}  \,du_j &=& 0,\\
-2 \eta_{1k}^\prime du_1 - 2 \eta_{2k}^\prime du_2 + 4 \sum \limits_{i=1}^2\sum\limits_{j=1}^2 \alpha_{ij}w_{ik}^\prime du_j&= &2 \pi i \, dv_k
\end{array}
\label{twoidentities}
\end{equation}
Since $du_1$ and $du_2$ are linearly independent, the first of the identities in Equation~(\ref{twoidentities}) implies
\begin{equation}
\begin{array}{rcl}
-2 \eta_{1k} + 4 \alpha_{11} \omega_{1k}+4 \alpha_{21} \omega_{2k} &= & 0, \\
-2 \eta_{2k} + 4 \alpha_{12} \omega_{1k} + 4 \alpha_{22} \omega_{2k} & = & 0.
\end{array}
\end{equation}
Since $\alpha_{12} = \alpha_{21},$ the preceding equations imply a key relation between the period matrices of the differentials $\{du_1, du_2\}$ and $\{dr_1,dr_2\}$ which will be essential in defining the Kleinian sigma function, viz.,
\begin{equation}
\alpha = \f{1}{2} \eta \omega^{-1}.
\end{equation}

\subsection{Riemann Theta Functions}
The Kleinian sigma function will be defined using the Riemann theta function associated with the period matrix $\tau$ of the normalized differentials $dv_1$ and $dv_2.$
\begin{definition}
The theta function of $v \in \mathbb{C}^2,$ associated with the period matrix $\tau,$ is 
\begin{equation}
\begin{array}{rcl}
\theta (v) &= &\sum\limits_{n_1 = -\infty}^{\infty}\sum\limits_{n_2 = -\infty}^{\infty} e^{2 \pi i (v_1 n_1+v_2 n_2)+\pi i (\tau_{11}n_1^2+ 2 \tau_{12}n_1n_2 + \tau_{22} n_2^2)} \\
&=& \sum\limits_{n= -\infty}^{\infty} e^{2 \pi i n v+\pi i n \tau n},
\end{array}
\end{equation}
with quasiperiodicity on the period lattice $\mathbb{C}^2/(I\bigoplus \tau)$ given by
\begin{equation}
\theta (v+m+\tau m^\prime) = e^{-2 \pi i m^{\prime}(v+\f{1}{2} \tau m^\prime)} \theta(v),
\end{equation}
where $m, m^\prime \in \mathbb{Z}^2.$ 
\end{definition}

\begin{definition}
The partial derivatives of $\theta$ are denoted by 
\begin{equation}
\begin{array}{rcl}
\theta_1 (v) &=& \f{\partial \theta}{\partial v_1} (v),\\[.1in]
\theta_2 (v) & = & \f{\partial \theta}{\partial v_2} (v).
\end{array}
\end{equation}
Moreover, the logarithmic derivatives of $\theta$ possess lattice transformations of the form, for $j=1,2,$
\begin{equation}
\f{\partial \ln \theta}{\partial v_j} (v+m + \tau m^\prime ) = -2 \pi i m_j^\prime + \f{\partial \ln \theta}{\partial v_j} (v).
\label{logtheta}
\end{equation}
\end{definition}

\begin{lemma}
The choice of the canonical cycles obeying the conjugate relations~(\ref{conjugatecycles}) implies that the real part of $\tau$ satisfies equation~(\ref{realtau}) and, hence, it can be shown that
\begin{equation}
\theta(v)^* = \theta(v^*).
\end{equation}
Similarly, for $i=1,2,$
\begin{equation}
\begin{array}{rcl}
\theta_1(v)^* & = & \theta_1(v^*),\\
\theta_2(v)^* & = & \theta_2 (v^*).
\end{array}
\end{equation}
\end{lemma}

\begin{theorem}
The theta function is zero at odd half-integer periods, viz., if $m, m^\prime \in \mathbb{Z}^2$ and $m m^\prime = m_1 m_1^\prime + 
m_2 m_2^\prime$ is an odd integer, then
$$\theta (\f{1}{2} m + \f{1}{2} \tau m^\prime) = 0.$$
\end{theorem}

\begin{proof}
\begin{equation}
\begin{array}{rcl}
\theta (\f{1}{2} m + \f{1}{2} \tau m^\prime) &= & \sum\limits_{n=-\infty}^{\infty} e^{\pi i mn+\pi i (n+\f{1}{2}m^\prime)\tau(n+\f{1}{2}m^\prime)-\f{\pi i}{4}m^\prime \tau m^\prime},\\
& = & \sum\limits_{n=-\infty}^{\infty} e^{-\pi i mn+\pi i (n-\f{1}{2}m^\prime)\tau(n-\f{1}{2}m^\prime)-\f{\pi i}{4}m^\prime \tau m^\prime},\\
& = &e^{-i\pi m m^\prime}\sum\limits_{n=-\infty}^{\infty} e^{-\pi i m(n-m^\prime)+\pi i (n-m^\prime+\f{1}{2}m^\prime)\tau(n-m^\prime+\f{1}{2}m^\prime)-\f{\pi i}{4}m^\prime \tau m^\prime},\\
& =& e^{-i\pi m m^\prime}\sum\limits_{n=-\infty}^{\infty}  e^{-\pi i m n+\pi i (n+\f{1}{2}m^\prime)\tau(n+\f{1}{2}m^\prime)-\f{\pi i}{4}m^\prime \tau m^\prime},\\
& =& e^{-i\pi m m^\prime}\sum\limits_{n=-\infty}^{\infty}  e^{-2\pi i m n +\pi i m n+\pi i (n+\f{1}{2}m^\prime)\tau(n+\f{1}{2}m^\prime)-\f{\pi i}{4}m^\prime \tau m^\prime},\\
& =& e^{-i\pi m m^\prime}\sum\limits_{n=-\infty}^{\infty}  e^{\pi i m n+\pi i (n+\f{1}{2}m^\prime)\tau(n+\f{1}{2}m^\prime)-\f{\pi i}{4}m^\prime \tau m^\prime},\\
& =& e^{-i\pi m m^\prime}\theta (\f{1}{2} m + \f{1}{2} \tau m^\prime), \\
\end{array}
\end{equation}
where the summation index is renamed, firstly, $n \rightarrow - n$ and, secondly, $n \rightarrow n-m^\prime,$ without change to the sum, which is over all possible integer pairs $n \in \mathbb{Z}^2.$ 
\end{proof}

\begin{definition}
The half-integer periods  are denoted as
\begin{equation}
\f{1}{2} \left[\begin{array}{cc} m_1^\prime & m_2^\prime \\ m_1 & m_2 \end{array} \right] = \f{1}{2} \left(\begin{array}{c}
m_1 + \tau_{11} m_1^\prime +\tau_{12} m_2^\prime\\ 
 m_2 + \tau_{21} m_1^\prime + \tau_{22} m_2^\prime) \end{array}\right) = \f{1}{2} m + \f{1}{2} \tau m^\prime,
\end{equation}
where $m, m^\prime \in \mathbb{Z}^2.$
\end{definition}

\begin{theorem}
The fifteen integrals 
\begin{equation}
v^{\lambda_i,\lambda_j} = \int_{\lambda_j}^{\lambda_i} dv,
\end{equation}
for $i, j = 1, \ldots, 6,$ with $i \neq j,$
of the normalized holomorphic differentials $dv = (dv_1, dv_2)^t$ on the dissected Riemann surface $\mathscr{R}$ constructed from the canonical homology cycles in Figure~\ref{cycles} between the fifteen pairs of distinct branch points of the Riemann surface $\mathscr{R}$ are equal to fifteen distinct non-zero half-integer periods, as follows,
\begin{equation}
\begin{array}{rclrclrcl}
v^{\lambda_1,\lambda_2} & = &\f{1}{2} \left[\begin{array}{cc} 0 & 0 \\ 1 & 0 \end{array} \right], & v^{\lambda_1, \lambda_3} & = &  \f{1}{2} \left[\begin{array}{cc} -1 & 1 \\0  & -1  \end{array} \right],        & v^{\lambda_1,\lambda_4} & = &   \f{1}{2} \left[\begin{array}{cc} -1 & 1 \\ 0 &0  \end{array} \right],       \\
v^{\lambda_1,\lambda_5} & = &\f{1}{2} \left[\begin{array}{cc} -1 &0  \\0  &-1  \end{array} \right],         & v^{\lambda_1, \lambda_6} & = &  \f{1}{2} \left[\begin{array}{cc} -1 &0  \\1  & 0 \end{array} \right]  ,      & v^{\lambda_2,\lambda_3} & = &   \f{1}{2} \left[\begin{array}{cc} -1 & 1 \\-1  &-1  \end{array} \right]  ,      \\
v^{\lambda_2,\lambda_4} & = &  \f{1}{2} \left[\begin{array}{cc}  -1& 1 \\-1  &0  \end{array} \right],       & v^{\lambda_2, \lambda_5} & = &  \f{1}{2} \left[\begin{array}{cc} -1 &0  \\  -1& -1 \end{array} \right],        & v^{\lambda_2,\lambda_6} & = &      \f{1}{2} \left[\begin{array}{cc} -1 &0  \\0  &0  \end{array} \right],     \\
v^{\lambda_3,\lambda_4} & = &  \f{1}{2} \left[\begin{array}{cc} 0 &0  \\ 0 & 1 \end{array} \right],       & v^{\lambda_3, \lambda_5} & = &  \f{1}{2} \left[\begin{array}{cc}  0&-1  \\0  &0  \end{array} \right],        & v^{\lambda_3,\lambda_6} & = &    \f{1}{2} \left[\begin{array}{cc} 0 &-1  \\1  &1  \end{array} \right],       \\
v^{\lambda_4,\lambda_5} & = &   \f{1}{2} \left[\begin{array}{cc} 0 &-1  \\0  & -1 \end{array} \right],      & v^{\lambda_4, \lambda_6} & = &    \f{1}{2} \left[\begin{array}{cc}  0&-1  \\1  &0  \end{array} \right],      & v^{\lambda_5,\lambda_6} & = &   \f{1}{2} \left[\begin{array}{cc} 0 &0  \\ 1 & 1 \end{array} \right].  \\
\end{array}
\end{equation}
\end{theorem}
\begin{proof}
By examining the dissection of the Riemann surface $\mathscr{R}$ given by Figure~\ref{cycles}, it is possible to integrate between any two branch points on the dissected Riemann surface while remaining on the lower sheet of the two-sheeted covering, without crossing any of the basis cycles.  However the same integral can be performed by tracing the same path projected on the upper sheet for which the integrand is the same except for multiplication by $-1$ due to the action of the hyperelliptic involution on $dv.$  By keeping track of the crossings of the homology cycles on the upper sheet (so as to remain on the dissected Riemann surface), the equality of the two integration procedures leads to the stated results.
\end{proof}
\begin{corollary} 
\begin{equation}
\theta ( v^{\lambda_1,\lambda_3} ) = \theta( v^{\lambda_1,\lambda_6})  =  \theta (v^{\lambda_2,\lambda_4})  =\theta ( v^{\lambda_2,\lambda_5} )  =    \theta( v^{\lambda_3,\lambda_6})  =\theta (v^{\lambda_4,\lambda_5})  =  0. 
\end{equation}
\end{corollary}
\begin{lemma}
If  $\theta (\int^\lambda_{\lambda_6} dv)$ is not identically zero, then it has precisely two simple zeros at $\lambda = \lambda_1, \lambda_3.$  
\end{lemma}
Note that $\theta(0)$ certainly does not vanish identically for all $\theta$ since $\theta (0) > 0$ when $\tau = \left(\begin{array}{cc} i & 0 \\ 0 & i \end{array} \right).$  In such cases $\theta (\int^\lambda_{\lambda_6} dv)$ is not identically zero, since it is not zero at $\lambda=\lambda_6.$
\begin{lemma}
If  $\theta (\int^\lambda_{\lambda_6} dv)$ is not identically zero and $z_1 \neq z_2,$ then
$\theta (\int_{\lambda_6}^\lambda dv - \int_{\lambda_1}^{z_1} dv - \int_{\lambda_3}^{z_2}dv)$ has precisely two simple zeros at $\lambda=z_1,z_2 \in \mathscr{R}.$
\end{lemma}
\begin{lemma}
If $\theta (\int^\lambda_{\lambda_6} dv) $ is not identically zero, then $\theta (z) = 0$  if and only if $\lambda \in \mathscr{R}$ such that
$$z = \int^\lambda_{\lambda_6} dv + \int^{\lambda_1}_{\lambda_3} dv.$$  The set of $z \in \mathbb{C}^2/(I \bigoplus \tau)$ having this property  is called the {\em theta divisor} $\Theta.$  The theta divisor is a one-complex-dimensional subvariety of the two-complex-dimensional period lattice $\mathbb{C}^2/(I \bigoplus \tau).$
\end{lemma}

\begin{lemma}
Suppose $\theta (\int^\lambda_{\lambda_6} dv)$ is not identically zero and $e \notin \Theta,$ then $\theta (\int^\lambda_{\lambda_6} dv - e)$  is not identically zero and has precisely two zeros $\lambda=z_1, z_2.$ Moreover, up to addition of integer multiples of periods,
\begin{equation}
e=\int_{\lambda_1}^{z_1} dv + \int_{\lambda_3}^{z_2}dv.
\label{zdefinition}
\end{equation}
\label{elemma}
\end{lemma}
Thus, with the exception of the one-complex-dimensional variety $e \in \Theta,$ the points $(z_1,s_1), (z_2,s_2) \in \mathscr{R}$ that satisfy Equation~(\ref{zdefinition}) may be viewed as well-defined functions of the independent variable $e \in \mathbb{C}^2/(I\bigoplus \tau).$

\subsection{Kleinian elliptic functions}
\begin{definition}
The fundamental Kleinian sigma function $\sigma$ of $u \in \mathbb{C}^2$ is 
\begin{equation}
\sigma(u) = e^{\f{1}{2} u \eta \omega^{-1} u} \theta(\f{1}{2} \omega^{-1} u),
\label{sigmadefinition}
\end{equation}
with transformations on the period lattice $\mathbb{C}^2/(2\omega \bigoplus 2\omega^\prime)$ given by, for $r=1,2,$
\begin{equation}
\begin{array}{ccc}
\sigma(u_1+2\omega_{1r},u_2+2\omega_{2r})=e^{2\eta_{1r}(u_1+\omega_{1r})+2\eta_{2r}(u_2+\omega_{2r})}\sigma(u),\\
\sigma(u_1+2\omega^\prime_{1r},u_2+2\omega^\prime_{2r})=e^{2\eta^\prime_{1r}(u_1+\omega^\prime_{1r})+2\eta^\prime_{2r}(u_2+\omega^\prime_{2r})}\sigma(u).
\end{array}
\end{equation}
In general, if $m = (m_1,m_2)^T$ and $m^\prime = (m_1^\prime,m_2^\prime)^T$ are two couples of integers and
\begin{equation}
\Omega_m = 2\omega m + 2 \omega^\prime m^\prime, \hspace{.2in} H_m = 2 \eta m + 2 \eta^\prime m^\prime,
\end{equation}
then
\begin{equation}
\sigma (u + \Omega_m) = e^{H_m (u + \f{1}{2} \Omega_m) - i \pi m m^\prime} \sigma (u).
\label{sigmalattice}
\end{equation} 
\end{definition}
Note that the sigma function can be multiplied by a constant factor independent of $u$ without altering any of the results of this paper.  The simplest normalization sufficient to accomplish the Jacobi inversion has been chosen. The definition of the Kleinian sigma function is motivated by the following identities between integrals on the Riemann surface.

Firstly, Equation~(\ref{alphadefinition}) implies
\begin{equation}
\begin{split}
\mathscr{K}(\lambda,a;z_1,b_1)+\mathscr{K}(\lambda,a;z_2,b_2) - \int_a^\lambda d\Pi_{z_1,b_1}-\int_a^\lambda d\Pi_{z_2,b_2} =\\
	-2 \sum\limits_{i=1}^2 \sum\limits_{j=1}^2 \alpha_{ij}\int^\lambda_a   du_i \left(\int_{b_1}^{z_1} du_j + \int_{b_2}^{z_2} du_j \right).
\end{split}
\label{fun1}
\end{equation}

Secondly, the identity
\begin{equation}
\begin{split}
\int_a^\lambda d\Pi_{z_1,b_1} +\int_a^\lambda d\Pi_{z_2,b_2} = \\
\log \f{\theta (\int_{\lambda_6}^\lambda dv - \int_{\lambda_1}^{z_1} dv - \int_{\lambda_3}^{z_2}dv) \theta (\int_{\lambda_6}^a dv - \int_{\lambda_1}^{b_1} dv - \int_{\lambda_3}^{b_2}dv)}{\theta (\int_{\lambda_6}^\lambda dv - \int_{\lambda_1}^{b_1} dv - \int_{\lambda_3}^{b_2}dv)\theta (\int_{\lambda_6}^a dv - \int_{\lambda_1}^{z_1} dv - \int_{\lambda_3}^{z_2}dv)},
\end{split}
\label{fun2}
\end{equation}
follows from considering the function
$$\f{\theta (\int_{\lambda_6}^\lambda dv - \int_{\lambda_1}^{z_1} dv - \int_{\lambda_3}^{z_2}dv)}{\theta (\int_{\lambda_6}^\lambda dv - \int_{\lambda_1}^{b_1} dv - \int_{\lambda_3}^{b_2}dv)} \exp \left(-\int_a^\lambda d\Pi_{z_1,b_1} -\int_a^\lambda d\Pi_{z_2,b_2}\right),$$
which must be constant because it has no zeros and no poles and zero periods around all homology basis cycles.  Therefore it is equal to its value when $\lambda = a.$   A similar conclusion follows even if only the periods around $\mathfrak{a}_1$ and $\mathfrak{a}_2$ are zero and the periods around $\mathfrak{b}_1$ and $\mathfrak{b}_2$ are constant (but not necessarily zero).

The two identities in Equations~(\ref{fun1}) and~(\ref{fun2}) can be written in terms of the Kleinian sigma function~(\ref{sigmadefinition}), using the fact that $\alpha = \f{1}{2} \eta \omega^{-1},$ as
\begin{equation}
\begin{split}
\mathscr{K}(\lambda,a;z_1,b_1)+\mathscr{K}(\lambda,a;z_2,b_2) = \\
\log \f{\sigma (\int_{\lambda_6}^\lambda du - \int_{\lambda_1}^{z_1} du - \int_{\lambda_3}^{z_2}du) \sigma (\int_{\lambda_6}^a du - \int_{\lambda_1}^{b_1} du - \int_{\lambda_3}^{b_2}du)}{\sigma (\int_{\lambda_6}^\lambda du - \int_{\lambda_1}^{b_1} du - \int_{\lambda_3}^{b_2}du)\sigma (\int_{\lambda_6}^a du - \int_{\lambda_1}^{z_1} du - \int_{\lambda_3}^{z_2}du)},
\end{split}
\label{fun3}
\end{equation}
where $du = (du_1,du_2)^T$ and $dv = (dv_1,dv_2)^T.$

Now Equation~(\ref{twoformidentity}), in which we interchange $(\lambda,a)$ and $(z,b)$ using the symmetry of the expression, and Equation~(\ref{Kdefinition}) imply
\begin{equation}
\begin{split}
\mathscr{K}(\lambda,a;z_1,b_1)+\mathscr{K}(\lambda,a;z_2,b_2) = 
\int_{b_1}^{z_1} \int_a^\lambda \f{\partial}{\partial \lambda} \f{s+w}{z-\lambda} \f{d\lambda}{2 s} dz +\\ \int_{b_2}^{z_2} \int_a^\lambda \f{\partial}{\partial \lambda} \f{s+w}{z-\lambda} \f{d\lambda}{2 s} dz +
\left(\int_{b_1}^{z_1} du_1 + \int_{b_2}^{z_2} du_1\right) \int_a^\lambda dr_1 + \\ \left(\int_{b_1}^{z_1} du_2 + \int_{b_2}^{z_2} du_2\right) \int_a^\lambda dr_2.
\end{split}
\label{fun4}
\end{equation}

The identities in Equations~(\ref{fun3}) and~(\ref{fun4}) combine to give
\begin{equation}
\begin{split}
\log \f{\sigma (\int_{\lambda_6}^\lambda du - \int_{\lambda_1}^{z_1} du - \int_{\lambda_3}^{z_2}du) \sigma (\int_{\lambda_6}^a du - \int_{\lambda_1}^{b_1} du - \int_{\lambda_3}^{b_2}du)}{\sigma (\int_{\lambda_6}^\lambda du - \int_{\lambda_1}^{b_1} du - \int_{\lambda_3}^{b_2}du)\sigma (\int_{\lambda_6}^a du - \int_{\lambda_1}^{z_1} du - \int_{\lambda_3}^{z_2}du)} = \\
\int_{b_1}^{z_1} \int_a^\lambda \f{\partial}{\partial \lambda} \f{s+w}{z-\lambda} \f{d\lambda}{2 s} dz +\int_{b_2}^{z_2} \int_a^\lambda \f{\partial}{\partial \lambda} \f{s+w}{z-\lambda} \f{d\lambda}{2 s} dz +\\
\left(\int_{b_1}^{z_1} du_1 + \int_{b_2}^{z_2} du_1\right) \int_a^\lambda dr_1 + \left(\int_{b_1}^{z_1} du_2 + \int_{b_2}^{z_2} du_2\right) \int_a^\lambda dr_2.
\end{split}
\label{fun5}
\end{equation}

Lemma~\ref{elemma} and the discussion following Equation~(\ref{fun2}) shows that in Equation~(\ref{fun5}), for $u^\prime, u^{\prime\prime} \notin \Theta,$ we can write
$$u^\prime = \int_{\lambda_1}^{z_1} du + \int_{\lambda_3}^{z_2}du, \hspace{.2in} u^{\prime\prime} = \int_{\lambda_1}^{b_1} du + \int_{\lambda_3}^{b_2}du,$$
consider $z_1$ and $z_2$ as functions of $u^\prime$ and $b_1$ and $b_2$ as functions of $u^{\prime\prime}.$
Thus Equation~(\ref{fun5}) becomes,
\begin{equation}
\begin{split}
\log \f{\sigma (\int_{\lambda_6}^\lambda du - u^\prime) \sigma (\int_{\lambda_6}^a du - u^{\prime\prime})}{\sigma (\int_{\lambda_6}^\lambda du - u^{\prime\prime})\sigma (\int_{\lambda_6}^a du - u^\prime)} = 
\int_{b_1}^{z_1} \int_a^\lambda \f{\partial}{\partial \lambda} \f{s+w}{z-\lambda} \f{d\lambda}{2 s} dz + \\ \int_{b_2}^{z_2} \int_a^\lambda \f{\partial}{\partial \lambda} \f{s+w}{z-\lambda} \f{d\lambda}{2 s} dz +
\left(u^\prime_1-u^{\prime\prime}_1\right) \int_a^\lambda dr_1 + \left(u^\prime_2-u^{\prime\prime}_2\right) \int_a^\lambda dr_2.
\end{split}
\label{fun6}
\end{equation}
Differentiation of Equation~(\ref{fun6}) with respect to $u_i^\prime$ for $i=1,2,$  produces
\begin{equation}
\begin{split}
-\zeta_i \left(\int^\lambda_{\lambda_6}du -u^\prime \right) +\zeta_i \left(\int^a_{\lambda_6} du - u^\prime \right) = 
\f{\partial z_1}{\partial u_i^\prime} \int_a^\lambda \f{\partial}{\partial \lambda} \f{s_1+w}{z_1-\lambda} \f{d\lambda}{2 s_1}
+ \\ \f{\partial z_2}{\partial u_i^\prime} \int^\lambda_a \f{\partial}{\partial \lambda} \f{s_2+w}{z_2 -\lambda} \f{d\lambda}{2 s_2}
+\int^\lambda_{a} dr_i,
\end{split}
\label{zetaeqn}
\end{equation}
where the zeta functions $\zeta_i$  and $\wp_{ij}$ functions for $i,j=1,2,$ are defined by
\begin{equation}
\zeta_i (u) = \f{\partial}{\partial u_i} \log \sigma (u), \hspace{.2in} \wp_{ij} = - \f{\partial^2}{\partial u_i \partial u_j} \log \sigma (u).
\end{equation}
Equation~(\ref{sigmalattice}) shows that, for $i,j = 1,2,$ and $\Omega_m$ an integer translation across the period lattice of the Jacobi variety of $\mathscr{K}_2,$
\begin{equation}
\zeta_i (u + \Omega_m) - \zeta_i (u) = (H_m)_i = 2 \eta_{i1}m_1+2\eta_{i2} m_2 + 2\eta^\prime_{i1} m_1^\prime + 2 \eta_{i2}^\prime m_2^\prime
\end{equation}
and
\begin{equation}
\wp_{ij} (u + \Omega_m) = \wp_{ij} (u).
\end{equation}
Using 
\begin{equation}
\begin{array}{rcl}
u_1 & = & \int^{z_1}_{\lambda_1} \f{dz}{s} + \int^{z_2}_{\lambda_3} \f{dz}{s}, \\
u_2 & = & \int^{z_1}_{\lambda_1} \f{z dz}{s} + \int^{z_2}_{\lambda_3} \f{z dz}{s},
\end{array}
\end{equation}
we find that
\begin{equation}
\begin{array}{rclrcl}
\f{\partial z_1}{\partial u_1} & = & \f{s_1 z_2}{z_2-z_1}, &\hspace{.2in} \f{\partial z_2}{\partial u_1} &= & \f{s_2 z_1}{z_1-z_2},\\
\f{\partial z_1}{\partial u_2} & = & \f{s_1}{z_1-z_2}, & \hspace{.2in} \f{\partial z_2}{\partial u_2} & = & \f{s_2}{z_2-z_1}.
\end{array}
\label{zderivatives}
\end{equation}
In Equation~(\ref{zetaeqn}) we now make the change of replacing the points $(z_1,s_1), (z_2,s_2) \in \mathscr{K}_2$ with their corresponding points under the
hyperelliptic involution, viz.,  $(z_1,-s_1)$ and $(z_2,-s_2).$  Notice that $u_r^\prime$ is changed by this transformation to
\begin{equation}
-u_r^\prime +  2 \omega_{i1}m_1+2\omega_{i2} m_2 + 2\omega^\prime_{i1} m_1^\prime + 2 \omega_{i2}^\prime m_2^\prime
\end{equation}
for some integers $m_1,m_2,m_1^\prime,m_2^\prime.$  Making the corresponding change in the right-hand side of Equation~(\ref{zetaeqn}) and using Equation~(\ref{zderivatives}), we obtain, for $i=1,2,$
\begin{equation}
\int^\lambda_a dr_i + \zeta_i \left( \int^\lambda_{\lambda_6} du + u^\prime \right) - \f{1}{2} f_i (\lambda, z_1,z_2) =  \zeta_i \left( \int^a_{\lambda_6} du + u^\prime \right) - \f{1}{2} f_i (a, z_1,z_2),
\label{newzetaeqn}
\end{equation}
where
\begin{equation}
\begin{array}{rcl}
f_1 (\lambda,z_1,z_2) & = & \f{w (\lambda - z_1-z_2)}{(z_1-\lambda)(z_2-\lambda)} + \f{s_1 (z_1 - \lambda - z_2)}{(z_1-\lambda)(z_1-z_2)} +\f{s_2 (z_2-\lambda-z_1)}{(z_2-\lambda)(z_2-z_1)}, \\
f_2 (\lambda,z_1,z_2) & = & \f{w}{(z_1-\lambda)(z_2-\lambda)} + \f{s_1}{(z_1-\lambda)(z_1-z_2)} + \f{s_2}{(z_2-\lambda)(z_2-z_1)}.
\end{array}
\label{feqn}
\end{equation}
By adding $\int^{z_1}_{\lambda_1} dr_i + \int^{z_2}_{\lambda_3} dr_i$ to both sides of Equation~(\ref{newzetaeqn}), we obtain
\begin{equation}
\begin{split}
\int^\lambda_a dr_i + \int^{z_1}_{\lambda_1} dr_i + \int^{z_2}_{\lambda_3} dr_i + \zeta_i \left( \int^\lambda_{\lambda_6} du + u^\prime \right) - \f{1}{2} f_i (\lambda, z_1,z_2) \\
 = \int^{z_1}_{\lambda_1} dr_i + \int^{z_2}_{\lambda_3} dr_i + \zeta_i \left( \int^a_{\lambda_6} du + u^\prime \right) - \f{1}{2} f_i (a, z_1,z_2),
\end{split}
\label{symmetriczetaeqn}
\end{equation}
in which the left-hand side is symmetric with respect to $\lambda, z_1, z_2,$ and the right-hand side is the value of the left-hand side when $\lambda=a.$  Consequently the left-hand side of Equation~(\ref{symmetriczetaeqn}) is independent of $\lambda, z_1,$ and $z_2.$  Therefore, for $i=1,2,$
\begin{equation}
\zeta_i \left( \int^\lambda_{\lambda_6} du + u^\prime  \right) 
=  C_i + \f{1}{2} f_i (\lambda,z_1,z_2) - \int^\lambda_a dr_i - \int^{z_1}_{\lambda_1} dr_i - \int^{z_2}_{\lambda_3} dr_i ,
\label{Ceqn}
\end{equation}
where $C_i$ is independent of $\lambda, z_1, z_2.$  Since $a$ is an arbitrary point on $\mathcal{K}_2,$ it can be set equal to the branch point at $a=\lambda_6,$ giving
\begin{equation}
\zeta_i \left( \int^\lambda_{\lambda_6} du + u^\prime  \right) 
=  C_i + \f{1}{2} f_i (\lambda,z_1,z_2) - \int^\lambda_{\lambda_6} dr_i - \int^{z_1}_{\lambda_1} dr_i - \int^{z_2}_{\lambda_3} dr_i ,
\label{Ceqn}
\end{equation}
Now $\zeta_i(0) =0,$ being an odd function, so setting $\lambda = \lambda_6, z_1 = \lambda_1,$ and $z_2 = \lambda_3,$ shows that $C_i = 0,$ for $i=1,2.$

Direct calculation shows that, as  $\lambda \rightarrow \infty^{\pm},$ the singularities in the terms on the right-hand side of equation~(\ref{Ceqn}) cancel out, so that
\begin{equation}
\lim_{\lambda \rightarrow \infty^{\pm}} \frac{1}{2} f_i (\lambda,z_1,z_2) - \int^\lambda_{\lambda_6} dr_i = \gamma_i^{\pm},
\label{limiteqn}
\end{equation}
where 
\begin{equation}
\begin{array}{rcl}
\gamma_1^{\pm} & = & \pm \f{1}{2} (-z_1 z_2 -\f{1}{8} \Lambda_1^2+\f{1}{2}\Lambda_2) + \f{1}{2}\f{s_1-s_2}{z_1-z_2} + \delta_1^\pm,\\[.1in]
\gamma_2^{\pm} & = & \pm \f{1}{2} (z_1 + z_2 -\f{1}{2} \Lambda_1) + \delta_2^\pm,
\end{array}
\end{equation}
and $\delta_1^\pm$ and $\delta_2^\pm$ are constants independent of $z_1$ and $z_2.$ 
Thus equation~(\ref{Ceqn}) becomes
\begin{equation}
\zeta_i \left(\int^{\infty^\pm}_{\lambda_6} du + u^\prime \right) = \gamma_i^\pm -\int^{z_1}_{\lambda_1} dr_i - \int^{z_2}_{\lambda_3} dr_i.
\label{newCeqn}
\end{equation}
Equation~(\ref{newCeqn}), implies, for $i=1,2,$
\begin{equation}
\begin{split}
\zeta_i \left(u^\prime + \int^{\infty^+}_{\lambda_6} du \right) - \zeta_i \left(u^\prime + \int^{\infty^-}_{\lambda_6} du \right) = \gamma_i^+ - \gamma_i^-
\end{split}
\end{equation}
Substituting the explicit expressions for $\gamma^\pm_i,$ we obtain the solution to the Jacobi inversion problem,
\begin{equation}
\begin{array}{rcl}
\zeta_1 \left(u^\prime + \int^{\infty^+}_{\lambda_6}du \right) - \zeta_1 \left(u^\prime + \int^{\infty^-}_{\lambda_6} du \right) &=&-z_1 z_2 -\f{1}{8} \Lambda_1^2+\f{1}{2} \Lambda_2 + \delta_1,\\[.2in]
\zeta_2 \left(u^\prime + \int^{\infty^+}_{\lambda_6} du \right) - \zeta_2 \left(u^\prime + \int^{\infty^-}_{\lambda_6} du \right) &=&z_1+z_2 -\f{1}{2} \Lambda_1 + \delta_2,
\end{array}
\label{implicitsolution}
\end{equation}
where the constants $\delta_1$ and $\delta_2$ are independent of $z_1$ and $z_2,$ and so can be obtained from equation~(\ref{implicitsolution}) by setting 
$z_1=(\lambda_1,0)$ and $z_2 =(\lambda_2,0)=(\lambda_1^*,0).$  
Recall $\alpha_{ij}$ for $i,j = 1,2,$ is given by
\begin{equation}
\alpha = \f{1}{2} \eta \omega^{-1}  = \left(\begin{array}{cc} \alpha_{11} & \alpha_{12} \\[.1in] \alpha_{21} & \alpha_{22} \end{array} \right).
\end{equation} Hence 
\begin{equation}
\begin{array}{rcl}
\delta_1=\delta_1^+-\delta_1^- & = &  - (- |\lambda_1|^2 - \f{1}{8} \Lambda_1^2 + \f{1}{2} \Lambda_2) \\
 & &+ \zeta_1 \left(\int^{\infty^+}_{\lambda_6}du +\int^{\lambda_2}_{\lambda_3} du \right)- \zeta_1 \left(\int^{\infty^-}_{\lambda_6}du +\int^{\lambda_2}_{\lambda_3} du \right)\\[.1in]
& = &  - (- |\lambda_1|^2 - \f{1}{8} \Lambda_1^2 + \f{1}{2} \Lambda_2) + 2 \alpha_{11} \int^{\infty^+}_{\infty^-} du_1 + (\alpha_{12}+\alpha_{21}) \int^{\infty^+}_{\infty^-} du_2\\
&& +\f{1}{2}(\omega^{-1})_{11} (\f{\theta_1^+}{\theta^+}-\f{\theta_1^-}{\theta^-})+\f{1}{2}  (\omega^{-1})_{21} (\f{\theta_2^+}{\theta^+} - \f{\theta_2^-}{\theta^-}),\\[.1in] 
\delta_2 =\delta_2^+-\delta_2^- & = & - (\lambda_1 + \lambda_1^* -\f{1}{2} \Lambda_1)\\
& &+\zeta_2 \left(\int^{\infty^+}_{\lambda_6}du +\int^{\lambda_2}_{\lambda_3}\right)- \zeta_2 \left(\int^{\infty^-}_{\lambda_6}du +\int^{\lambda_2}_{\lambda_3} du\right)\\
& =&  - (\lambda_1 + \lambda_1^* -\f{1}{2} \Lambda_1) +(\alpha_{12}+\alpha_{21})\int^{\infty^+}_{\infty^-} du_1 + 2 \alpha_{22} \int^{\infty^+}_{\infty^-} du_2\\
& & +\f{1}{2} (\omega^{-1})_{12}(\f{\theta_1^+}{\theta^+}-\f{\theta_1^-}{\theta^-})+\f{1}{2} (\omega^{-1})_{22}(\f{\theta_2^+}{\theta^+}-\f{\theta_2^-}{\theta^-}),
\end{array}
\label{deltaeqn}
\end{equation}
where
\begin{equation}
\begin{array}{rcl}
\theta^+ & = & \theta \left(\f{1}{2} \omega^{-1}( \int^{\infty^+}_{\lambda_6} du + \int^{\lambda_2}_{\lambda_3} du)\right),\\[.1in]
\theta^- & = & \theta \left(\f{1}{2} \omega^{-1}( \int^{\infty^-}_{\lambda_6} du + \int^{\lambda_2}_{\lambda_3} du)\right),\\[.1in]
\theta_1^+ & = & \theta_1 \left(\f{1}{2} \omega^{-1}( \int^{\infty^+}_{\lambda_6} du + \int^{\lambda_2}_{\lambda_3} du)\right),\\[.1in]
\theta_1^- & = & \theta_1 \left(\f{1}{2} \omega^{-1}( \int^{\infty^-}_{\lambda_6} du + \int^{\lambda_2}_{\lambda_3} du)\right),\\[.1in]
\theta_2^+ & = & \theta_2 \left(\f{1}{2} \omega^{-1}( \int^{\infty^+}_{\lambda_6} du + \int^{\lambda_2}_{\lambda_3} du)\right),\\[.1in]
\theta_2^- & = & \theta_2 \left(\f{1}{2} \omega^{-1}( \int^{\infty^-}_{\lambda_6} du + \int^{\lambda_2}_{\lambda_3} du)\right).\\[.1in]
\end{array}
\end{equation}
Explicit expressions for $z_1$ and $z_2$ can then be found from the quadratic formula for the roots of the quadratic equation
\begin{equation}
z^2 - (z_1+z_2) z + z_1 z_2 = 0.
\end{equation}
The corresponding values of $s_1$ and $s_2$ for $(z_1,s_1), (z_2,s_2) \in \mathscr{K}_2$ are, for $k=1,2,$
\begin{equation}
\begin{split}
s_k = \left( \wp_{22} \left(u^\prime + \int^{\infty^-}_{\lambda_6} du \right) - \wp_{22}\left(u^\prime + \int^{\infty^+}_{\lambda_6} du \right) \right) z_k \\+ 
\wp_{12} \left(u^\prime + \int^{\infty^-}_{\lambda_6} du \right) - \wp_{12}\left(u^\prime + \int^{\infty^+}_{\lambda_6} du \right).
\end{split}
\end{equation}

\end{document}